\documentclass[12pt,a4paper]{article}
\usepackage{xcolor,paralist}
\usepackage[pdftex,colorlinks=true,hypertexnames=false]{hyperref}
\definecolor{darkblue}{rgb}{0,0,.6}
\hypersetup{citecolor=darkblue,linkcolor=darkblue,urlcolor=darkblue}
\usepackage{amssymb,enumitem,booktabs,subfig,setspace,bm,longtable,dsfont,orcidlink}
\usepackage[final,nopatch=footnote]{microtype}
\usepackage[top=1in, bottom=1in, left=1in, right=1in]{geometry}
\usepackage[authoryear]{natbib}
\usepackage[thinc]{esdiff}
\usepackage{comment,mathtools,longtable}
\usepackage{float,fancyvrb}
\usepackage{url}
\usepackage{booktabs}

\usepackage{makecell} 
\usepackage{amsmath}
\usepackage{amsfonts}
\usepackage{epsfig}
\usepackage{graphics}
\usepackage[normalem]{ulem}
\usepackage{palatino,eulervm}
\usepackage{graphicx}
\usepackage{lscape}
\usepackage{rotating}
\usepackage[linewidth=1pt]{mdframed}
\allowdisplaybreaks[4]
\DeclareMathOperator*{\argmin}{arg\,min}

\providecommand{\U}[1]{\protect\rule{.1in}{.1in}}

\graphicspath{{plots/}}
\newsavebox\CBox

\usepackage{amsthm,thmtools}
\usepackage{mathrsfs}

\makeatletter
\def\th@newremark{\th@remark\thm@headfont{\bfseries}}
\makeatletter
\theoremstyle{newremark}

\newtheorem{prop}{Proposition}

\declaretheoremstyle[spaceabove=6pt, spacebelow=6pt, headfont=\bfseries, notefont=\mdseries, notebraces={(}{)}, bodyfont=\normalfont, postheadspace=0.5em]{mystyle}

\newcommand{\X}{\mathcal{X}}
\newcommand{\Y}{\mathcal{Y}}

\newcommand{\Z}{\mathcal{Z}}
\bibpunct{(}{)}{;}{a}{,}{ and }

\makeatletter
\newcommand*{\addFileDependency}[1]{
\typeout{(#1)}
\@addtofilelist{#1}
\IfFileExists{#1}{}{\typeout{No file #1.}}
}\makeatother

\newcommand{\Rlogo}{\protect\includegraphics[height=1.8ex,keepaspectratio]{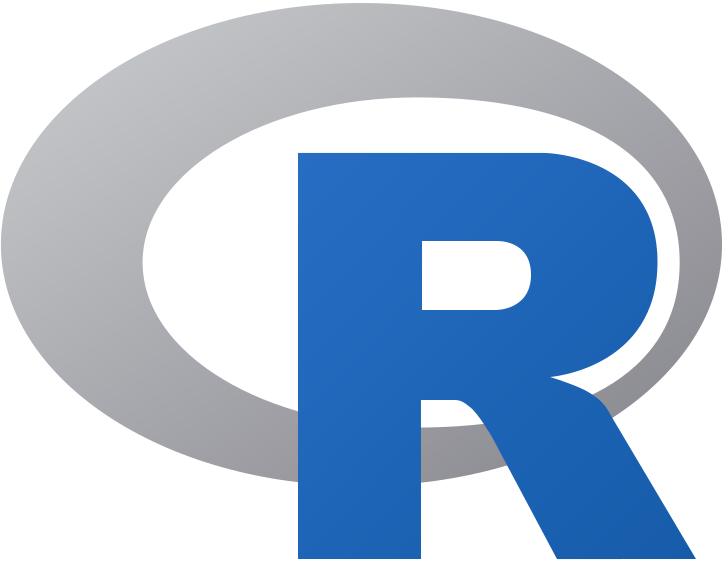}}

\begin{document}

\title{Interpretable models for forecasting high-dimensional functional time series}

    \author{\normalsize Han Lin Shang \orcidlink{0000-0003-1769-6430} \\
\normalsize Department of Actuarial Studies and Business Analytics \\
\normalsize Macquarie University \\
\\
\normalsize  Cristian F. Jim\'{e}nez-Var\'{o}n \orcidlink{0000-0001-7471-3845}\footnote{Department of Mathematics, Ian Wand Building, University of York, Deramore Lane, York YO10 5GH, United Kingdom; Telephone: +44 (0)1904 321629; Email: cristian.jimenezvaron@york.ac.uk}\\
\normalsize Department of Mathematics \\
\normalsize University of York
}

\date{}

\maketitle

\centerline{\bf Abstract}
We study the modeling and forecasting of high-dimensional functional time series, which can be temporally dependent and cross-sectionally correlated. Central to our implementation is a functional analysis of variance by decomposing high-dimensional functional time series, such as subnational age- and sex-specific mortality observed over years, into two distinct components: a deterministic mean structure and a residual process varying over time. Unlike purely statistical dimensionality-reduction techniques, the functional analysis of variance decomposition provides an interpretable framework by partitioning the series into effects attributable to data-specific factors, such as regional and sex-level variations, and a grand functional mean. From the residual process, we implement a functional factor model to capture the remaining stochastic trends. By combining the forecasts of the residual component with the estimated deterministic structure, we obtain the forecasted curves for high-dimensional functional time series. Illustrated by the age-specific Japanese subnational mortality rates from 1975 to 2023, we evaluate and compare the accuracy of the point and interval forecasts across various forecast horizons. The results demonstrate that leveraging these interpretable components not only clarifies the underlying drivers of the data, but also improves point forecast accuracy by about 25\% to 45\% compared to an existing method, providing more transparent insights for evidence-based policy decisions, such as accurate modeling of financial costs of length of stay in the old-aged care facilities.

\vspace{.02in}

\newpage
\setstretch{1.5}

\section{Introduction}\label{sec:1}

Recent advances in computer storage and recording facilitate the presence of functional data across several scientific fields, from cosmic demographics \citep[see, e.g.,][]{LBK+25} to small-area official estimation \citep[see, e.g.,][]{MSZ16}. When a functional variable is measured over time, it gives rise to functional time series analysis \citep[see][for a comprehensive review]{HK12, KR17}. When one observes multiple functional variables, the crux of the problem is to capture temporal dependence: each series has its own temporal dependence; together, they also exhibit cross-sectional dependence. That gives rise to high-dimensional functional time series (HDFTS) \citep[see, e.g.,][]{GSY19, CFQ+25}, also known as functional panel data \citep[see, e.g.,][]{JSS25}. Applications include temperature curves in hundreds of weather stations \citep[see, e.g.,][]{Delaigle2019}; return curves in finance are typically available for hundreds of stocks \citep[see, e.g.,][]{LLS+26}. The existing methods, such as factor-based dimensionality reduction approaches, work well in forecast accuracy, they lack  interpretability, an issue that we intend to address in this paper.

HDFTS can be expressed as $\bm{\Z}_{t} = (\Z_{1t},\dots,\Z_{Nt})^{\top}$, $i=1,\dots, N$, $t=1,\dots,T$, with $\Z_{it} = (\Z_{it}(u), u\in \mathcal{C}_i)\in \mathcal{H}_i$, where $\mathcal{H}_i$ is a Hilbert space defined as a set of measurable and square-integrable functions on a compact set $\mathcal{C}_i$. A distinct feature of HDFTS is that the number of cross sections can exceed the number of curves, i.e., $N>T$.In our example, $N = 47\times 2$, which is larger than the sample size $T=49$.

In the HDFTS literature, \cite{ZD23} derived Gaussian and multiplier bootstrap approximations for the sums of HDFTS. Using these approximations, they constructed joint simultaneous confidence bands for the mean functions and developed a hypothesis test to assess whether the mean functions in the panel dimension exhibit parallel behavior. \cite{TSY22} studied clustering HDFTS, while \cite{LLS24} proposed hypothesis tests for the detection and estimation of change points, and further clustering of common change points in HDFTS using an information criterion. With a two-stage representation, \cite{HG18} applied a univariate functional principal component analysis to extract a set of principal components and their associated scores; by stacking the scores, one can further reduce dimensionality through principal component analysis. In the same vein, \cite{GSY19} stacked the principal component scores across all populations by their eigenvalue orders; using a factor model, the dimensionality can be further reduced. The above works focus on statistical inference; we aim to contribute a new modeling technique for HDFTS and to adapt a predictive inference framework, namely conformal prediction, for constructing pointwise prediction intervals in this paper. 

Relating to our work, \cite{HNT23} investigated the representation of HDFTS using a factor model, identifying conditions on the eigenvalues of the covariance operator crucial for establishing the existence and uniqueness of the factor model. In estimating HDFTS models, \cite{TNH23} developed a functional factor model with functional factor loadings and a matrix of real-valued factors. With real-valued factor loadings and functional factors, \cite{GQW+26} studied a different factor representation for decomposing HDFTS. By unifying both factor models, \cite{LLS+26} considered functional factor loadings and functional factors. 

While these factor-based approaches are computationally efficient for dimensionality reduction, a significant challenge in their widespread application is that the resulting factors and loading matrices are often difficult to interpret \citep{Liu2026}. In many high-dimensional settings, the latent factors are purely mathematical constructs that do not directly correspond to observable categorical drivers, such as geographic regions or demographic cohorts. This ``black-box'' nature can limit the utility of the model for practitioners who demand an understanding of the specific sources of variation, such as sex-specific trends or regional disparities, that drive the overall series. In our procedure, the functional ANOVA is first deployed to exactly decompose the HDFTS into various mean terms and residuals.

The FANOVA models are not new, since they are particularly useful for analyzing data across a range of applications, such as human tactile perception \citep{SME03}, menstrual cycle data \citep{BR98}, and circadian rhythms with random effects and smoothing-spline analysis of variance decomposition \citep{WKB03}. \cite{KS10} established a Bayesian framework for FANOVA modeling to estimate the effect of geographical regions on Canadian temperatures. \cite{SG12} proposed a functional median polish modeling as an extension of the univariate median polish of \cite{Tukey77}, and a functional rank test was developed to determine the significance of the effects of the functional main factor. While \cite{Shang25} considered a one-way FANOVA for modeling and forecasting a time series of Lorenz curves, \cite{JSS24} proposed a two-way FANOVA to model and forecast subnational age-specific mortality observed over time.

We introduce a structural approach that prioritizes parameter interpretability by singling out the observed categorical drivers without sacrificing the flexibility of factor modeling. Following earlier work by \cite{JSS24}, we consider a two-way functional analysis of variance (FANOVA) to decompose an HDFTS. The FANOVA models evaluate the functional effects of categorical variables (known as factors) by determining how functions differ across their levels. Differing from \cite{JSS24}, our two-way FANOVA incorporates an interaction term that accounts for dependence between the two categorical variables, namely gender and prefecture. This interaction term is subsequently modeled via a one-way FANOVA, thereby providing a more flexible representation of cross-effect relationships than existing approaches. 

Regarding our age- and sex-specific subnational mortality rates, our modeling framework consists of three steps:
\begin{inparaenum}[(I)]
\item we begin with the two-way FANOVA to model mean patterns associated with region and sex; 
\item the interaction term in the two-way FANOVA can be captured by one-way FANOVA; 
\item the sex-specific residuals across regions over time from the one-way FANOVA are then modeled through a functional factor model.
\end{inparaenum}
In addition, we want to examine whether or not the interaction term in step~(II) can enhance point and interval forecast accuracy.

In this paper, we advance the HDFTS literature by proposing an interpretable framework that separates deterministic and stochastic variation. We first apply a two-way FANOVA to decompose HDFTS into data-driven functional components associated with observable factors, such as region and sex. We then model group-specific interactions through a one-way FANOVA and capture the remaining time-varying dynamics using a functional factor model. By combining the interpretability of FANOVA with the flexibility of functional factor models, the proposed approach provides clear insights into cross-group heterogeneity. An application to Japanese age- and sex-specific subnational mortality data demonstrates that the method improves both point and interval forecast accuracy relative to existing approaches.

The remainder of the paper is organized as follows. Section~\ref{sec:2} introduces the data and provides a preliminary analysis. Section~\ref{sec:3} presents the proposed decomposition framework, including the two-way and one-way FANOVA and a functional factor model. Section~\ref{sec:4} outlines the construction of prediction intervals. Section~\ref{sec:5} discusses model fitting, while Section~\ref{sec:6} evaluates forecasting performance. Section~\ref{sec:7} provides concluding remarks and discusses future research directions.

\section{Japanese subnational age-specific mortality rates}\label{sec:2}

In many developed countries, such as Japan, increasing longevity and population aging have raised concerns about the sustainability of pension, health and aged-care systems \citep[see, e.g.,][]{Coulmas07}. These challenges have intensified the need for accurate age-specific mortality modeling and forecasting. While national mortality forecasts are important, subnational forecasts are essential for informing regional policy and allocating current and future resources, such as aged-care facilities. Despite their value for capturing regional heterogeneity, subnational mortality data are often of poor quality.

Obtained from \cite{JMD26}, we investigate subnational age- and sex-specific mortality rates in Japan from 1975 to 2023. We acknowledge that the COVID-19 years (2020-2023) may cause a structural shock to mortality patterns and affect forecast accuracy for those years. Due to data availability, this shock can gradually be observed with updated data.

Mortality rates are the ratios of registered death counts to population exposure in the relevant year for the given age (based on a one-year age group). We examine age groups ranging from 0 to 94 in a single year of age, and the last age group including all ages at and above~95. In Figures~\ref{fig:1a} and~\ref{fig:1b}, we present a rainbow plot of age- and sex-specific mortality in Okinawa, the most southern prefecture in Japan. Although a decreasing trend can be observed, it is evident that the noise level of such a series is quite high.
\begin{figure}[!htb]
\centering
\subfloat[Raw data]
{\includegraphics[width=8.56cm]{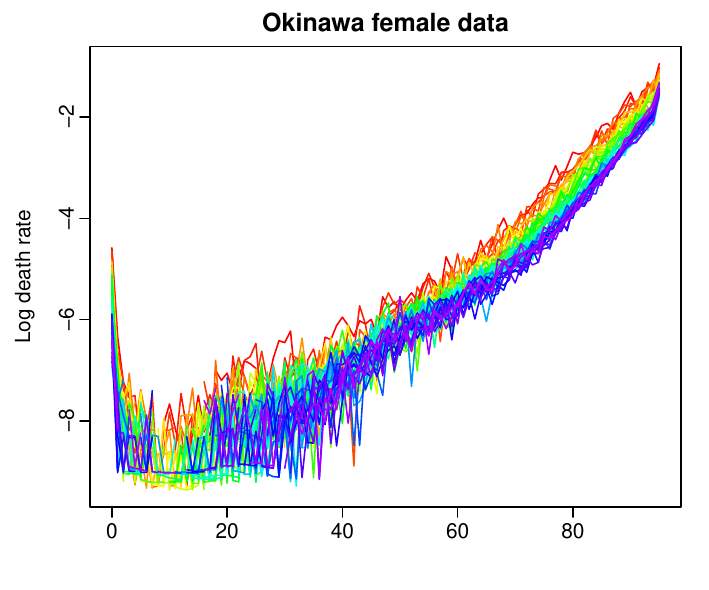}\label{fig:1a}}
\quad
\subfloat[Raw data]
{\includegraphics[width=8.56cm]{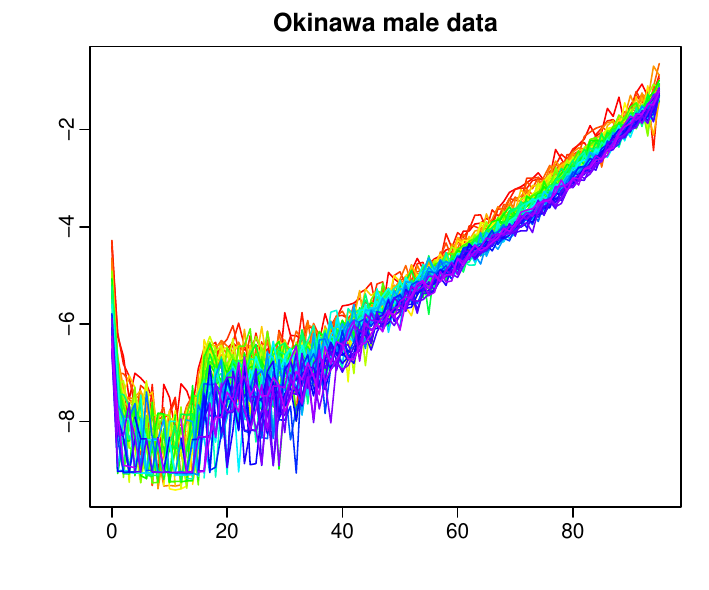}\label{fig:1b}}
\\
\subfloat[Smoothed data]
{\includegraphics[width=8.56cm]{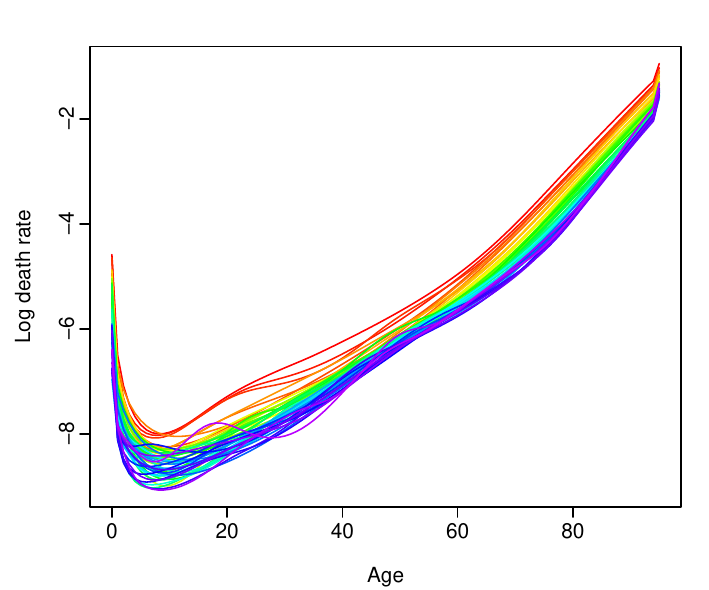}\label{fig:1c}}
\quad
\subfloat[Smoothed data]
{\includegraphics[width=8.56cm]{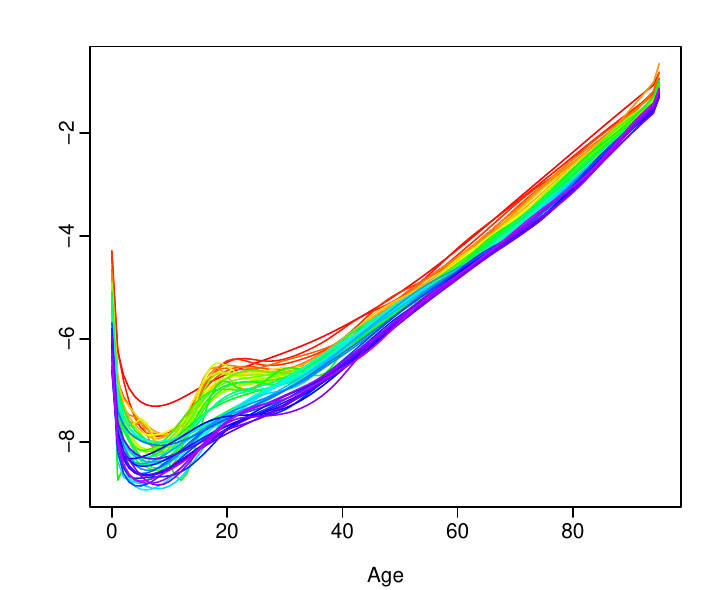}\label{fig:1d}}
\caption{\small Age-specific raw and smoothed $\log$ mortality rates for ages 0 to 95+ between 1975 and 2023 in Okinawa. Curves are ordered chronologically according to the colors of the rainbow. The oldest years are shown in red, with the most recent years in violet. The left vertical axis measures log mortality rates. Especially for the male series}, mortality rates dip in early childhood, climb in the teen years, stabilize in the early 20s, and then steadily increase with age. \label{fig:1}
\end{figure}

These multiple functional time series are observed with error at times $t=1, 2,\dots,T$ and we aim to forecast the functions for times $t=T+1,\dots, T+h$, where $h$ denotes a forecast horizon. Let ${\Z^g_{it}(u_j)}$ denote mortality rates observed discretely at ages $(u_1,\dots,u_p)$ for gender $g=\{F,M\}$. We assume underlying smooth $L_2$ functions ${\Y^g_{it}(u_j)}$ such that
\begin{equation*}
\ln{\Z^g_{it}(u_j)} = \ln{\Y^g_{it}(u_j)} + \sigma^g_t(u_j)\varepsilon^g_{it,j},\qquad j=1, 2,\dots,p.
\end{equation*}
where $\ln(\cdot)$ represents the natural logarithm, $\{\varepsilon^g_{it,j}\}$ are independent and identically distributed variables with zero mean and unit variance, and $\sigma^g_t(u_j)$ allows heteroskedasticity; for different ages, the variance term may be different. The smoothing step helps improve point forecast accuracy. Since conformal prediction is model-agnostic, it is insensitive to whether the original or smoothed log mortality rates are used. Due to the increasing monotonicity in mortality beyond age 65, we consider a $P$-spline with monotonic constraint \citep{Wood94}. Computationally, the \verb|smooth.demogdata| function in the demography package with the default parameters in \Rlogo \citep{Hyndman25} is used. This function is capable of modeling sparsely observed data.

In Figures~\ref{fig:1c} and~\ref{fig:1d}, we display smoothed $\log$ mortality rates for Okinawa. Not only do smoothing techniques help address missing values and mortality rates that exceed one, but they also enhance forecast accuracy \citep[see, e.g.,][]{HS09}.

\section{Decomposition of high-dimensional functional time series}\label{sec:3}

\subsection{Two-way functional analysis of variance (TWA)}\label{sec:3.1}

Through a two-way FANOVA, smoothed log mortality rates, $\ln{\Y^g_{it}(u_j)}$, can be decomposed as 
\begin{equation}
\ln{\Y^g_{it}(u)} = \mu(u) + \alpha_i(u) + \delta^g(u) + \gamma_{ig}(u) + \X_{it}^g(u),\label{eq:1}
\end{equation}
where $\Y^{g}_{it}(u)$ denotes the smoothed mortality rate at age $u$ for gender $g$ in prefecture $i=1, 2,\dots,N$. In this model, $\mu(u)$ represents the functional grand effect, $\alpha_{i}(u)$ denotes the functional row (prefecture) effect, $\delta^{g}(u)$ signifies the functional column (gender) effect, $\gamma_{ig}(u)$ denotes the interaction term, and $\X_{it}^{g}(u)$ represents the residual component. By grouping the interaction term and the residual component, we denote $\X_{it}^{g,\dagger}(u) = \gamma_{ig}(u) + \X_{it}^g(u)$.

The functional grand, row, and column effects can be estimated using their sample means \citep[see also Chapter 13,][]{RS06}. They are given by
\begin{align*}
\widehat{\mu}(u) &= \frac{1}{N\times 2\times T}\sum^{N}_{i=1}\sum^2_{g=1}\sum^T_{t=1}\ln \Y_{it}^{g}(u), \\
\widehat{\alpha}_i(u) &= \frac{1}{2\times T}\sum^2_{g=1}\sum^T_{t=1}\ln \Y_{it}^g(u) - \widehat{\mu}(u), \\
\widehat{\delta}^g(u) &= \frac{1}{N\times T}\sum^{N}_{i=1}\sum^T_{t=1}\ln \Y_{it}^g(u) - \widehat{\mu}(u).
\end{align*}
To ensure identifiability, several constraints are in place:
\begin{align}
\sum^{N}_{i=1}\alpha_i(u) &= \sum^2_{g=1}\delta^g(u) = 0, \\
\sum^{N}_{i=1}\X_{it}^{g,\dagger}(u) &= \sum^2_{g=1}\X_{it}^{g,\dagger}(u) = 0,\quad \forall t.
\label{eq:constraints}
\end{align}

The estimates $\widehat{\mu}(u)$, $\widehat{\alpha}_{i}(u)$, and $\widehat{\delta}^g(u)$ capture the overall mean, the regional mean, and the gender mean, $\X_{it}^{g,\dagger}(u)$ is time-varying and itself an HDFTS. Since females and males have different mortality profiles, we sequentially adopt a one-way FANOVA in Section~\ref{sec:3.2} or a functional factor model of \cite{LLS+26} in Section~\ref{sec:3.3} to capture any remaining regional pattern.

\subsection{One-way functional analysis of variance (OWA)}\label{sec:3.2}

The interaction term from the two-way FANOVA decomposition states that the prefecture effect influences each gender differently. To estimate the interaction term, we employ a one-way FANOVA decomposition \citep[see, e.g.,][]{Shang25}. 

Within our HDFTS framework, the observed mortality curves $\ln{\Y^g_{it}(u)}$ are grouped by two factors at the same time: a regional factor (prefectures) and a demographic factor (gender). Because every curve belongs to a specific prefecture-gender combination, an interaction effect may exist, in which the difference between male and female mortality shapes varies across different geographical regions. In Section~\ref{sec:3.1}, we grouped this fixed interaction structure with the temporal errors to form the mixed residual series 
\begin{equation*}
\X^{g,\dagger}_{it}(u) = \gamma_{ig}(u) + \X_{it}^{g}(u). 
\end{equation*}
To systematically separate this fixed prefecture-gender difference from the year-to-year changes, we apply a sequential one-way FANOVA directly onto the series $\X^{g,\dagger}_{it}(u)$.

In a standard functional one-way ANOVA layout \citep[see, e.g., Equation 13.1 in][]{RS06}, a model typically includes a functional grand mean, say $\theta^g(u)$, to center the data. However, because our data has already been completely centered across both factors by the first-stage TWA, this grand mean is mathematically forced to be zero. We formalize this cross-stage identifiability property in the following proposition.

\begin{prop}[Sequential identifiability of the TWA--OWA decomposition]
\label{prop:identifiability}
Let $\X_{it}^{g,\dagger}(u)$ denote the residual functional process obtained after the first-stage two-way FANOVA decomposition in Equation~\eqref{eq:1}, satisfying the identifiability constraints
\[
\sum_{i=1}^{N}\X_{it}^{g,\dagger}(u)=0,
\qquad t=1,\ldots,T.
\]
Consider applying a second-stage one-way FANOVA decomposition separately for each gender,
\begin{equation*}
\X_{it}^{g,\dagger}(u)
=
\theta^g(u)
+
\eta_i^g(u)
+
\mathcal{R}_{it}^g(u),
\end{equation*}
where the components are estimated by the corresponding sample means. Then the sequential TWA--OWA decomposition preserves identifiability in the following sense:

\begin{enumerate}
\item[(i)]
The second-stage functional grand mean is identically zero,
\[
\theta^g(u)\equiv 0,
\qquad g\in\{\mathrm{F},\mathrm{M}\}.
\]
Hence, the second-stage OWA does not introduce an additional location component.

\item[(ii)]
The extracted prefecture-specific interaction profiles satisfy the zero-sum constraint,
\[
\sum_{i=1}^{N}\eta_i^g(u)=0,
\qquad g\in\{\mathrm{F},\mathrm{M}\}.
\]

\item[(iii)]
The remaining stochastic functional process retains the cross-sectional centering property,
\[
\sum_{i=1}^{N}\mathcal{R}_{it}^g(u)=0,
\qquad t=1,\ldots,T,
\quad
g\in\{\mathrm{F},\mathrm{M}\}.
\]
\end{enumerate}
\end{prop}
\begin{proof}
Proof of this result is in Appendix~\ref{proof:prop1}.
\end{proof}

By using Proposition~\ref{prop:identifiability} to drop the redundant grand mean, we avoid over-parameterization and guarantee absolute model identifiability. The second-stage OWA framework decomposes the mortality patterns that exist for specific genders in specific regions (like Tokyo-Males vs. Tokyo-Females vs. Okinawa-Females) after we remove the general national averages as follows:
\begin{equation}
\X^{g,\dagger}_{it}(u) = \eta^g_{i}(u) + \mathcal{R}^{g}_{it}(u) \quad \text{for}\quad g\in\{\text{F}, \text{M}\},\label{eq:2}
\end{equation}
where $\eta_i^g(u)$ represents the row effect corresponding to the subnational gender-prefecture interaction mean, and $\mathcal{R}^g_{it}(u)$ denotes the remaining stochastic residual component. 

The reason for extracting the fixed interaction profile $\eta_i^g(u)$ via Equation~\eqref{eq:2} before running the functional factor model is straightforward. By filtering out these stable, baseline regional differences in the gender gap, we prevent constant cross-sectional shifts from mixing into the downstream functional factor model described in Section~\ref{sec:3.3}. Without this sequential OWA step, these fixed regional variations would contaminate the factor loadings and factors, causing systematic bias in our long-term mortality forecasts.

Because the nested errors $\mathcal{R}^g_{it}(u)$ strictly sum to zero across all prefectures for every single year~$t$ (as proven in Proposition~\ref{prop:identifiability}), no information is confounded or cross-contaminated between stages. Empirically, the interaction curve $\eta^g_{i}(u)$ in \eqref{eq:2} can be approximated by
\begin{align*}
\widehat{\eta}^{g}_{i}(u) &= \frac{1}{T}\sum^T_{t=1}\X_{it}^{g,\dagger}(u).
\end{align*}

\subsection{A functional factor model (FFM)}\label{sec:3.3}

It is recognized that HDFTS are influenced by common functions over the temporal dimension, leading to strong cross-sectional dependence. For example, age- and sex-specific mortality curves collected across different regions may be influenced by common patterns in each region. Hence, the functional factor model plays an important role in modeling HDFTS. In the statistical literature, \cite{HNT23} and \cite{TNH23} present a factor model with functional factor loadings and scalar-valued factors. \cite{GQW+26} propose a factor model with scalar-valued factor loadings and functional factors. By unifying both factor models, \cite{LLS+26} introduce a factor model with functional factor loadings and functional factors, which we adopt.

For each gender $g$, we consider the following functional factor model:
\begin{equation}
\mathcal{R}^g_{it}(u) = \sum_{k=1}^{k_*}\int_{C_{k}^{*}} B^g_{ik}(u,v)F^g_{tk}(v)dv+\varepsilon^g_{it}(u), \quad u\in C_i, \label{eq:3}
\end{equation}
where $F^g_{tk}(v)$ represents time-varying factors, $B^g_{ik}(u,v)$ is a  continuous two-dimensional linear operator representing the $k$\textsuperscript{th} factor loading for the $i$\textsuperscript{th} prefecture, gender $g$ at two ages $u$ and $v$, $k_{*}$ is the number of factors, $C_{k}^{*}$ represents a compact set which may be different from $C_i$, and $\varepsilon^g_{it}(u)$ is the error term of the model. By imposing a low-dimensional functional factor condition on the latent factor $F^g_{tk}(v)$, we obtain the series approximation
\begin{equation}
F^g_{tk}(v) = \Phi^g_k(v)^{\top}\bm{G}^g_t + \eta^g_{tk}(v), \quad v\in C_{k}^{*}, \quad k=1, 2,\dots,k_{*}, \label{eq:4}
\end{equation}
where $\Phi^{g}_{k}(v)$ is a $q$-dimensional vector of basis functions, $\bm{G}^g_t$ is a $q$-dimensional vector of 
random variables, and $\eta^g_{tk}(v)$ is the approximation error. Note that the type of basis functions is not important, but $F^g_{tk}(v)$ has a low-rank representation that absorbs into the representation of $\mathcal{R}^g_{it}(u)$.

By plugging~\eqref{eq:4} into~\eqref{eq:3}, we obtain
\begin{align*}
\mathcal{R}^g_{it}(u) &= \sum_{k=1}^{k_*}\int_{\mathcal{C}_k^{*}}B^g_{ik}(u,v)\left[\Phi^g_k(v)^{\top}\bm{G}^g_t + \eta^g_{tk}(v)\right]dv+\varepsilon^g_{it}(u) \\
&=\sum^{k_*}_{k=1}\int_{\mathcal{C}_k^{*}} B^g_{ik}(u,v)\Phi^g_k(v)^{\top}dv\bm{G}^g_t+\sum^{k_*}_{k=1}\int_{\mathcal{C}_k^{*}} B^g_{ik}(u,v)\eta^g_{tk}(v)dv + \varepsilon^{g}_{it}(u) \\
&=\Lambda^g_i(u)^{\top}\bm{G}^g_t + \varepsilon_{it}^{*g}(u),
\end{align*}
where $\Lambda^g_i(u) = \sum^{k_*}_{k=1}\int_{\mathcal{C}_k^{*}} B^g_{ik}(u,v)\Phi_k^g(v)^{\top}dv$ and $\varepsilon^{*g}_{it}(u) = \sum_{k=1}^{k_*}\int_{\mathcal{C}_k^{*}} B^g_{ik}(u,v)\eta^g_{tk}(v)dv + \varepsilon^g_{it}(u)$. Note that the notation $k_{*}$ is aggregated out in $\Lambda^g_i(u)$ and $\varepsilon^{*g}_{it}(u)$. Hereafter, the number of factors is determined by the dimension of $\bm{G}_{t}^{g}$, that is $q$.

To estimate $\widetilde{\bm{G}}^g_t$, we first estimate the covariance of $\mathcal{R}^g_{it}(u)$ by
\begin{equation*}
\bm{\Delta}^g = (\Delta^g_{tt^{'}})_{T\times T}\quad \text{with}\quad \Delta^g_{tt^{'}} = \frac{1}{N}\sum^{N}_{i=1}\int_{u\in \mathcal{C}_i} \mathcal{R}^g_{it}(u)\mathcal{R}^g_{it^{'}}(u)du.
\end{equation*}
By eigenanalysis of the $T\times T$ matrix $\bm{\Delta}^g$, we obtain $\widetilde{\bm{G}}^g=(\widetilde{\bm{G}}^g_{1},\dots,\widetilde{\bm{G}}^g_T)^{\top}$ as a $T\times q$ matrix with columns being the eigenvectors (multiplied by $\sqrt{T}$) corresponding to the $q$ largest eigenvalues of~$\bm{\Delta}^g$. The factor loading functions are estimated as
\begin{equation*}
\widetilde{\Lambda}^g_i(u) = \frac{1}{T}\sum^T_{t=1}\mathcal{R}^g_{it}(u)\widetilde{\bm{G}}^g_t,\qquad i=1,\dots,N,
\end{equation*}
via least squares, using the normalization restriction $\frac{1}{T}\sum^T_{t=1}\widetilde{\bm{G}}^g_t\widetilde{\bm{G}}_t^{g\top} = \bm{I}^g_q$. Note that, as is standard in factor modeling frameworks, the factors $\widetilde{\bm{G}}^g_t$ and factor loading functions $\widetilde{\Lambda}^g_i(u)$ are uniquely identified up to a sign change. However, because these components enter the system as an inner product, any arbitrary sign flip cancels out identically, leaving the reconstructed space $\widetilde{\Lambda}^g_i(u)^{\top}\widetilde{\bm{G}}^g_t$ and its downstream forecasts completely invariant.

To estimate the number of factors $q$, we resort to an information criterion introduced in \cite{LLS+26}, which can slowly diverge to infinity. Let $\nu_{\ell}(\bm{\Delta}^g/T)$ be the $\ell$\textsuperscript{th} largest eigenvalue of $\bm{\Delta}^g/T$ and define
\begin{equation*}
\widehat{q} = \argmin_{1\leq \ell\leq q_{\max}}\left[\nu_{\ell}(\bm{\Delta}^g/T)+\ell\phi_{NT}\right]-1,
\end{equation*}
where $\phi_{NT}$ is the penalty term and $q_{\max}$ is a user-specified positive integer. In practice, we set $\phi_{NT} = [\max(T,N)]^{-\frac{1}{2}}$ and $q_{\max} = T$.

Conditional on the estimated factor loadings $\widetilde{\Lambda}^g_i(u)$, the time-varying dynamics are captured by the $q$-dimensional factors $\widetilde{\bm{G}}^g_t$. For each factor, we apply a univariate ETS forecasting method to obtain the $h$-step-ahead forecast of $\widetilde{\bm{G}}^g_{T+h}$. The optimal ETS model is determined by an automatic algorithm based on the corrected Akaike information criterion in the forecast package. The optimal model does not need to be the same across prefectures and genders. For completeness, we also explore the point forecast accuracy based on autoregressive integrated moving average (ARIMA) model in Appendix~\ref{App:A}.

By multiplying the forecasted factors with the estimated factor loadings, the $h$-step-ahead forecast of $\mathcal{R}_{it}^g(u)$ is given by
\begin{equation*}
\widehat{\mathcal{R}}_{i(T+h)}^g(u) = \widetilde{\Lambda}^g_i(u)^{\top} \widehat{\widetilde{\bm{G}}}^g_{T+h},
\end{equation*}
where $\widehat{\widetilde{\bm{G}}}^g_{T+h}$ denotes the $h$-step-ahead forecast of $\widetilde{\bm{G}}^g_{T+h}$.

\section{Construction of prediction intervals}\label{sec:4} 

While we present a flexible approach to model and forecast HDFTS in Section~\ref{sec:3}, it is also important to quantify the forecast uncertainty, which is often measured using a statistical model. This approach may be vulnerable to model misspecification, selection bias, and limited finite-sample validity. While bootstrapping can potentially mitigate some of these concerns \citep[see, e.g.,][]{PS23}, it is often computationally demanding. Here, we take a model agnostic and distribution-free approach, namely conformal prediction, to construct prediction intervals in HDFTS. Among a rich family of conformal prediction methods, we consider split and sequential conformal predictions.

Specifically, for the TWA + OWA + FFM framework in \eqref{eq:5}, we first jointly model all prefectures and genders using two-way FANOVA to extract the deterministic grand and main effects. We then address the interaction term by applying a gender-specific one-way FANOVA, which decomposes the remaining variation into deterministic prefecture-specific means and a stochastic residual process $\mathcal{R}_{it}^g(u)$. By combining~\eqref{eq:1} with~\eqref{eq:2}, we obtain $h$-step-ahead point forecast of the HDFTS on the $\log$ scale as:
\begin{equation}
\ln{\widehat{\Y}}_{i, T+h}^g(u) = \underbrace{\widehat{\mu}(u) + \widehat{\alpha}_i(u) + \widehat{\delta}^g(u)}_{\text{Deterministic (TWA)}} + \underbrace{ \widehat{\eta}_i^g(u)}_{\text{Deterministic (OWA)}} + \underbrace{\widehat{\mathcal{R}}_{i, T+h}^g(u)}_{\text{Stochastic (FFM)}}.\label{eq:5}
\end{equation}
For the TWA + FFM framework, we assume that there is no deterministic interaction component ($ \eta_i^g(u) = 0$). Consequently, the construction of the prediction interval utilizes the residuals $\mathcal{X}_{it}^{g,\dagger}(u)$ obtained directly after the two-way FANOVA. 

In both modeling frameworks, we utilize split and sequential conformal prediction methods of \cite{Shang26}. In the split conformal prediction, we use a validation set to calibrate a tuning parameter $\xi_{\alpha, i}^g$ to scale pointwise summary measures (such as the standard deviation or absolute quantiles) of the residuals as described below. In the sequential conformal prediction, we avoid a fixed validation set by using an autoregressive model on the absolute residuals to update predictive quantiles as new data arrive. Finally, these calibrated stochastic components are added back to the deterministic structure to produce the complete interval forecasts.

\subsection{Split conformal prediction}\label{sec:4.1}

The split conformal prediction method uses a validation set to calibrate the empirical coverage probability to match the nominal coverage probability closely \citep[see also][]{ANH15}. We divide the 49-year sample (1975-2023) into training, validation, and test sets with proportions of 60\%, 20\%, and 20\%, respectively. This choice of proportions is arbitrary, but in line with the common practice in forecasting. Using the initial training data from 1975 to 2002, we implement an expanding-window forecasting scheme to generate $h$-step-ahead forecasts for the validation period from 2003 to 2013, for $h=1,2,\dots,10$. Under this forecasting scheme, the training sample is progressively enlarged at each iteration. The number of curves in the validation set or test set varies with the forecast horizon $h$. For instance, when $h=1$, there are 11 years to compute residual functions, corresponding to the differences between the observed curves in the validation set and their forecasts, whereas when $h=10$, there are only two years. From the absolute values of these residual functions, we compute pointwise summary measures $\vartheta_i^g(u)$, such as the pointwise standard deviation (abbreviated as Split (sd)) or the pointwise quantiles (abbreviated as Split (quantile)), for each region $i$ and gender~$g$. Alternatively, for quantile-based intervals, we take the absolute residuals and calculate the $100(1-\alpha)\%$ empirical quantiles, where $\alpha$ is the significance level, typically $\alpha=0.05$.

For a given forecast horizon $h$, let us denote the residual functions on the \textit{original} scale $\widehat{\epsilon}_{im}^g(u) = \Z_{im}^g(u) - \widehat{\Z}_{im}^g(u)$ for $m=1,\dots,M$, where $M$ denotes the number of years in the validation set. With the objective of minimizing the absolute difference between the empirical and nominal coverage probabilities, we seek to determine a tuning parameter $\xi_{\alpha, i}^g$ such that $100(1-\alpha)\%$ of the residual functions satisfy
\begin{equation*}
-\xi_{\alpha,i}^g\vartheta_i^g(u)\leq \widehat{\bm{\epsilon}}_{i}^g(u)\leq \xi_{\alpha,i}^g\vartheta_i^g(u),
\end{equation*}
where $\widehat{\bm{\epsilon}}_i^g(u)=[\widehat{\epsilon}_{i,1}^g(u),\dots,\widehat{\epsilon}_{i,M}^g(u)]$. Computationally, the \Verb|optim| function can be used. By the law of large numbers, when $M$ is reasonably large, one could achieve
\begin{equation*}
\mathbb{P}\left[-\xi_{\alpha,i}^g\vartheta_i^g(u)\leq \Z_{iT+h}^g(u) - \widehat{\Z}_{iT+h}^g(u)\leq \xi_{\alpha,i}^g\vartheta_i^g(u)\right]\approx \frac{1}{M}\sum^M_{m=1}\mathds{1}\left\{-\xi_{\alpha,i}^g\vartheta_i^g(u)\leq \widehat{\epsilon}_{im}^g(u)\leq \xi_{\alpha,i}^g\vartheta_i^g(u)\right\},
\end{equation*}
where $\Z_{iT+h}^g(u)$ denotes the holdout age-specific mortality rates in prefecture $i$ and gender $g$ in year $(T+h)$, $\widehat{\Z}_{iT+h}^g(u)$ denotes the corresponding forecasts in the original scale (after taking exponential back-transformation), and $\mathds{1}\{\cdot\}$ represents the binary indicator function. As a general guideline, $M$ should typically be on the order of tens.

\subsection{Sequential conformal prediction}\label{sec:4.2}

The split conformal prediction method requires a validation data set to calibrate a tuning parameter. Not only can it be time-consuming, but it also reduces the samples used to predict the test set. Without the need for a validation set, this sequential conformal prediction can automatically tune the predictive quantiles of the absolute residual functions as new data arrive. For example, with the last 10 years as the test set, we use the other years to compute the absolute residuals $\{|e_{i3}^g(u_j)|,\dots,|e_{i\iota}^g(u_j)|\}$ for a given age $u_j$. We require at least the first two curves to produce a forecast, so the residuals begin with the $3$\textsuperscript{rd} curve; that is, the first two out of 49 years data is negligible.

At the quantile of $1-\alpha$, we fit a quantile regression on lagged residuals, where the order of autoregression, denoted by $\mathrm{p}$, is determined by an information criterion, such as the Akaike information criterion \citep{Akaike1969}. Conditional on the most recent $\mathrm{p}$ number of absolute residuals as input, we predict a one-step-ahead quantile, denoted $\widehat{q}^g_{i\iota+1,\alpha}(u_j)$, where $\iota$ represents the data in the end of the training period. The prediction intervals are then given by 
\begin{equation*}
\widehat{\Z}_{i(\iota+1)}^g(u_j) \pm \widehat{q}^g_{i(\iota+1),\alpha}(u_j).
\end{equation*}
Once the actual curve $\Z_{i(\iota+1)}^g(u)$ arrives, we can update the absolute residual $|e_{i(\iota+1)}^g(u_j)|$ and refit.

\section{Model fitting}\label{sec:5}

\subsection{Model fitting via two-way functional analysis of variance}

Using subnational Japanese age- and sex-specific $\log$ mortality rates, we apply a two-way FANOVA to decompose HDFTS into different mean effects. In Figure~\ref{fig:2}, we display the grand mean effect, the row effect (prefecture), the column effect (sex), and the time-varying residuals that can contain the interaction term between prefecture and sex. The functional grand effect reveals an overall mean trend. The functional row effect reveals cross-regional heterogeneity, particularly for younger ages. The functional column effect shows a contrast between the male and female data.
\begin{figure}[!htb]
\centering
\subfloat[Functional grand effect]
{\includegraphics[width=5.71cm]{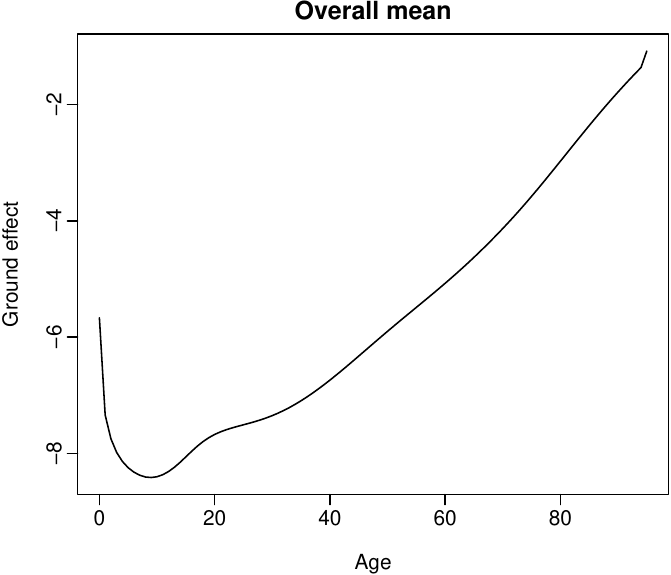}}
\quad
\subfloat[Functional row effect]
{\includegraphics[width=5.71cm]{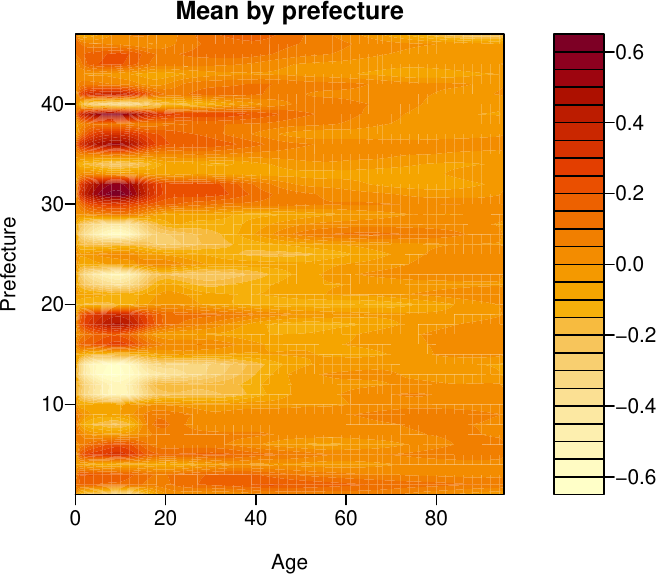}}
\quad
\subfloat[Functional column effect]
{\includegraphics[width=5.71cm]{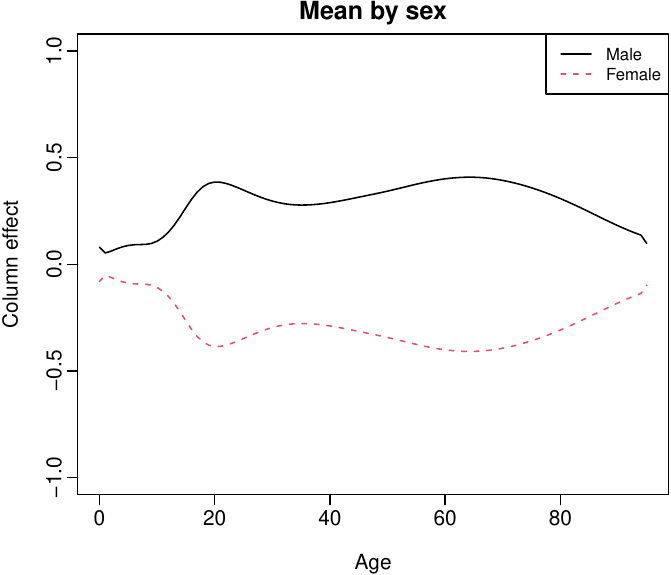}}
\caption{\small Two-way FANOVA decomposition: Functional grand effect, capturing the overall age profile; functional row effect, capturing the means across prefectures; and functional column effect, capturing the means across sexes.}\label{fig:2}
\end{figure}

In Japan, regional mortality heterogeneity is often larger at younger ages because deaths are more driven by external and behavioral causes (e.g., suicide, accidents), which vary strongly across prefectures. These causes are sensitive to local socioeconomic conditions, such as employment stability, rural depopulation, and social isolation. Epidemiologically, subnational variations in premature and youth mortality are heavily driven by localized disparities in major causes of death and socio-environmental stressors, which create highly divergent, non-senescent survival trajectories across different prefectures \citep{Tsuboi22}. Furthermore, the baseline shifts in these regional effects faithfully replicate Japan's documented subnational health variations and widening prefectural mortality gaps \citep{Nomura17}. Northeastern prefectures (e.g., Aomori) persistently exhibit elevated functional row effects, a pattern traditionally associated with severe winter climates, higher baseline cardiovascular disease risk profiles, and distinct dietary risk factor configurations (such as high sodium intake) typical of northern subnational divisions \citep{Nomura17}. Conversely, southwestern regions (e.g., Okinawa) reflect pronounced negative row deviations, matching their established profiles of exceptional historical longevity and traditional dietary configurations \citep{Willcox07}.

\subsection{Model fitting via one-way functional analysis of variance}\label{Results:OWA}

Like many developed nations, females generally have lower mortality rates than males at all ages in Japan \citep{Ikeda11}. Since females and males are biologically different \citep{Rogers2010}, we treat each of them separately. Thus, we apply one-way FANOVA to decompose the residual functions into a grand effect, a row effect, and the remainder term. 

In Figure~\ref{fig:3}, we present the decomposition terms obtained from the one-way FANOVA. For either females or males, the functional row effects reveal cross-region heterogeneity, more so for the male data.
\begin{figure}[!htb]
\centering
\subfloat[Functional row effect of the female data]
{\includegraphics[width=8.5cm]{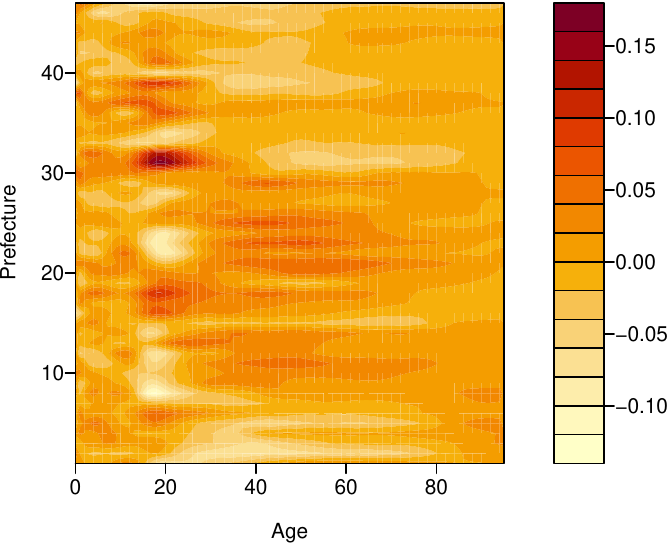}}
\qquad   
\subfloat[Functional row effect of the male data]
{\includegraphics[width=8.5cm]{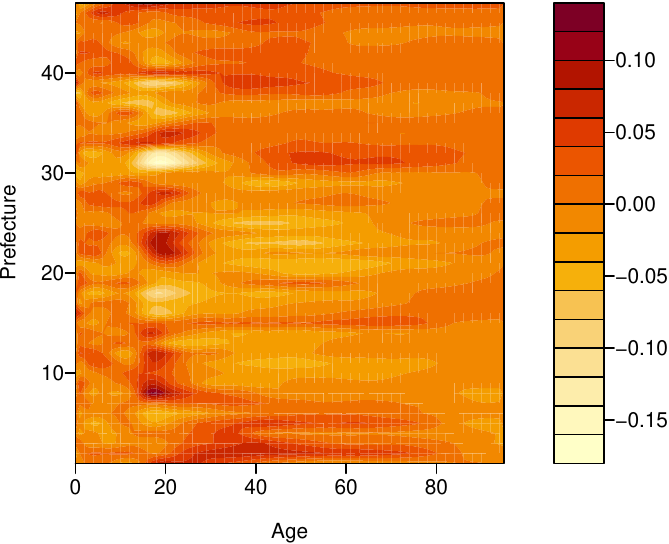}}
\caption{\small From the interaction term and time-varying residuals, we implement the one-way FANOVA for each gender to capture the row effect, respectively.}\label{fig:3}
\end{figure}

\subsection{Model fitting via functional factor model}

From the reminder term of the one-way FANOVA, we apply the functional factor model to extract the first set of factor loadings and its associated factors for the female and male data in Figure~\ref{fig:4}. While the two-way and one-way FANOVA models extract deterministic mean effects, the functional factor model aims to model the information associated with variance. 
\begin{figure}[!htb]
\centering
\subfloat[First functional factor loading of the female data]
{\includegraphics[width=8.5cm]{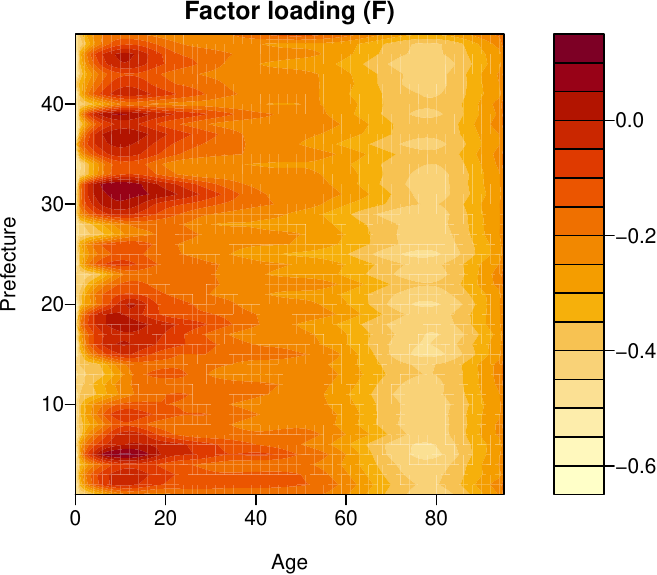}}
\qquad
\subfloat[First functional factor loading of the male data]
{\includegraphics[width=8.5cm]{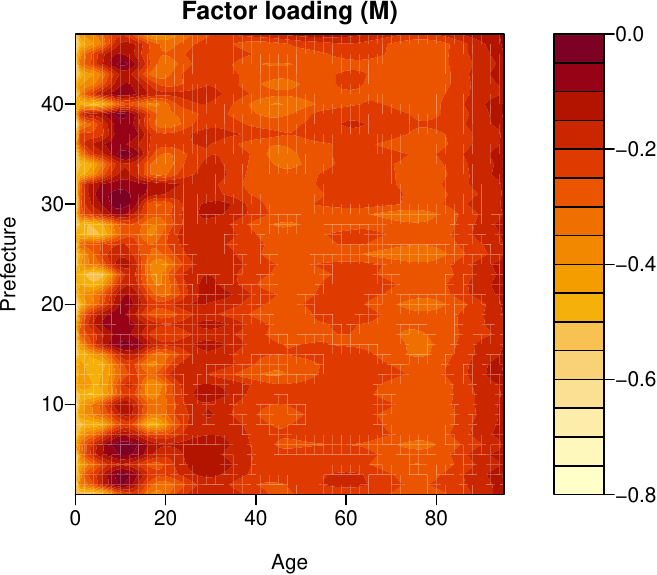}}
\\
\subfloat[First set of the factors of the female data]
{\includegraphics[width=8.5cm]{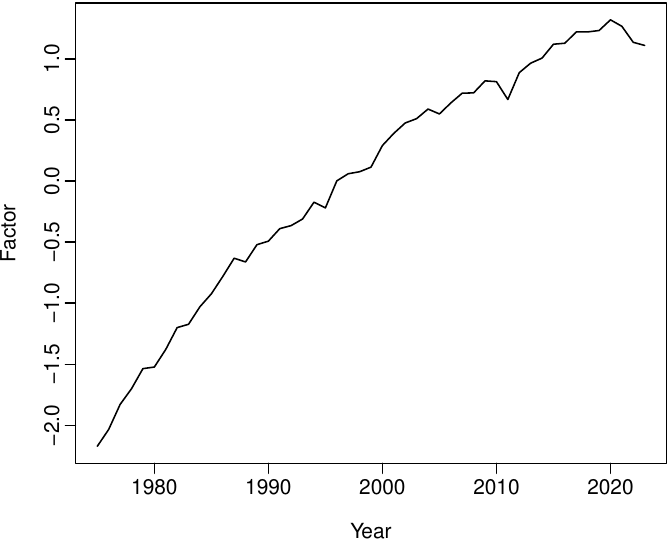}}
\qquad
\subfloat[First set of the factors of the male data]
{\includegraphics[width=8.5cm]{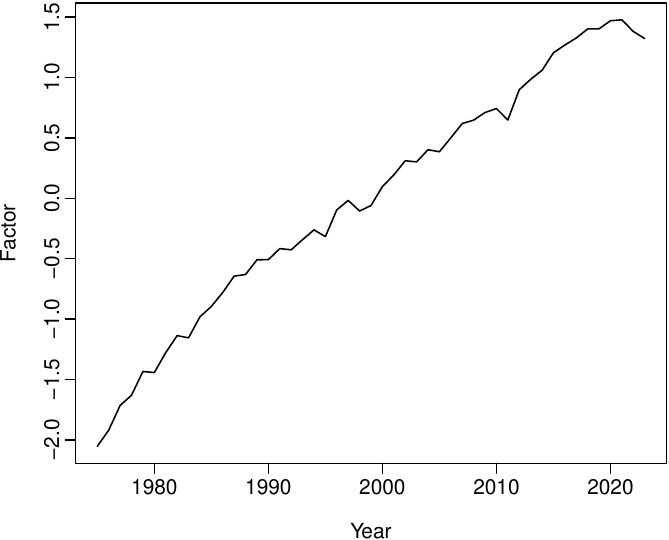}}
\caption{\small Via the functional factor model, we display the first set of factor loadings and factors, based on the estimated variance of the HDFTS residuals.}\label{fig:4}
\end{figure}

For the female and male series, there is an increasing time trend. Based on the estimated factor loadings, the greatest regional heterogeneity is observed in younger ages. Downstream, the functional factor model captures residual stochastic dynamics. Since the information criterion dominantly selects $\widehat{q}=1$ (as empirically documented across all expanding windows in Appendix~\ref{app:model_specs}), the primary functional factor represents the overarching nationwide secular decline in mortality achieved via universal improvements in public health and healthcare access \citep{Ikeda11}. The factor loadings thus serve as a proxy for regional convergence velocity, mapping which subnational economies adapt rapidly to national baseline shifts versus those experiencing distinct localized aging pressures.

\section{Results}\label{sec:6}

\subsection{Expanding-window forecast scheme}\label{Results:PFE}

An expanding window analysis is a robust method for evaluating the stability of model parameters and the accuracy of temporal predictions. This approach assesses structural consistency by iteratively computing parameter estimates and corresponding forecasts as the sample size grows \citep[see][pp. 313--314]{ZW06}. 

In this study, we utilize a training set consisting of the first 39 years (1975--2013) of Japanese subnational age- and sex-specific mortality data to generate one- to 10-step-ahead forecasts. We then re-estimate the time-series model parameters by incrementing the sample size by one year at each iteration until the full data period ending in 2023 is exhausted. This recursive procedure yields a total of 10 one-step-ahead forecasts, 9 two-step-ahead forecasts, and continues down to a single 10-step-ahead forecast. These values are compared with holdout observations to quantify out-of-sample accuracy. Figure~\ref{fig:expanding} illustrates this expanding window scheme for the horizon $h=1$, though horizons up to $h=10$ are evaluated.
\begin{figure}[!htb]
\begin{center}
\begin{tikzpicture}
\draw[->] (0,0) -- (10,0) node[right] {Time};
\draw[fill=blue!20] (0,-0.5) rectangle (3,0.5) node[midway] {Train};
\draw[fill=red!20] (3,-0.5) rectangle (3.5,0.5) node[midway] {F};
\draw[fill=blue!20] (0,-1.5) rectangle (5,-0.5) node[midway] {Train};
\draw[fill=red!20] (5,-1.5) rectangle (5.5,-0.5) node[midway] {F};
\draw[fill=blue!20] (0,-2.5) rectangle (7,-1.5) node[midway] {Train};
\draw[fill=red!20] (7,-2.5) rectangle (7.5,-1.5) node[midway] {F};
\draw[fill=blue!20] (0,-3.5) rectangle (9,-2.5) node[midway] {Train};
\draw[fill=red!20] (9,-3.5) rectangle (9.5,-2.5) node[midway] {F};
\node[left] at (0,0) {1975:2013};
\node[left] at (0,-1) {1975:2014};
\node[left] at (0,-2) {\hspace{-0.8in}{$\vdots$}};
\node[left] at (0,-3) {1975:2022};
\draw[fill=blue!20] (6.5,1) rectangle (7,1.5);
\node[right] at (7,1.25) {Training Window};
\draw[fill=red!20] (6.5,0.5) rectangle (7,1);
\node[right] at (7,0.75) {Forecast (F)};
\end{tikzpicture}
\end{center}
\caption{\small A diagram of the expanding-window forecast scheme. The data begin in 1975 and end in 2023.}\label{fig:expanding}
\end{figure}

\subsection{Point forecast error metrics}\label{sec:PFE}

To measure the point forecast errors, we use the root mean square forecast error (RMSFE) and the mean absolute forecast error (MAFE). For each region $i$ and gender $g$, the RMSFE and MAFE are 
\begin{align*}
\text{RMSFE}^g_i(h) &= \sqrt{\frac{1}{(11-h)\times 96}\sum^{11-h}_{\xi=h}\sum^{96}_{j=1}\left[\Z^g_{\eta+\xi,i}(u_j) - \widehat{\Z}^g_{\eta+\xi,i}(u_j)\right]^2},\\
\text{MAFE}^g_i(h) &= \frac{1}{(11-h)\times 96} \sum^{11-h}_{\xi=h}\sum^{96}_{j=1} \left|\Z^g_{\eta+\xi,i}(u_j) - \widehat{\Z}^g_{\eta+\xi,i}(u_j)\right|,
\end{align*}
where $\Z_{\eta+\xi,i}^g(u_j)$ represents the holdout sample for age $u_j$ and gender $g$ in region $i$, and $\widehat{\Z}_{\eta+\xi,i}^g(u_j)$ represents the corresponding point forecasts.

\subsection{Comparison of point forecast accuracy}\label{sec:CPF}

In Table~\ref{tab:1}, through a simple average of 47 prefectures, we report the one-step-ahead to ten-step-ahead point forecast errors for the holdout samples of the female and male data. In an early work of \cite{JSS24}, they considered the two-way FANOVA to capture the mean effects and then modeled the residuals by multivariate functional principal component analysis of \cite{SK22}. Compared to that method, our proposed methods based on FANOVA and the functional factor model improve the point forecast accuracy. Between the two frameworks, there is a slight preference for including the one-way FANOVA to capture the interaction term between prefecture and sex. 
\begin{center}
\renewcommand*{\arraystretch}{0.88}
\tabcolsep 0.195in
\begin{small}
\begin{longtable}{@{}ll ccc ccc @{}}
\caption{\small Averaged across 47 prefectures, we evaluate and compare the point forecast accuracy measured by RMSE and MAE. Forecasting method is ETS. The method with the smallest overall error is bolded. TWA + OWA + FFM represents our method combining two-way FANOVA, one-way FANOVA and a functional factor model. TWA + FFM represents our method combining two-way FANOVA and a functional factor model. TWA + MFTS is a benchmark method combining two-way FANOVA and multivariate functional time series method in \cite{JSS24}.}\label{tab:1} \\
\toprule
& & \multicolumn{3}{c}{Female} & \multicolumn{3}{c}{Male} \\
\cmidrule(lr){3-5} \cmidrule(lr){6-8}
Metric & $h$ & \makecell{TWA+\\OWA+FFM} & \makecell{TWA+\\FFM} & \makecell{TWA+\\MFTS} & \makecell{TWA+\\OWA+FFM} & \makecell{TWA+\\FFM} & \makecell{TWA+\\MFTS} \\ 
\midrule
\endfirsthead
\toprule
& & \multicolumn{3}{c}{Female} & \multicolumn{3}{c}{Male} \\
\cmidrule(lr){3-5} \cmidrule(lr){6-8}
Metric & $h$ & \makecell{TWA+\\OWA+FFM} & \makecell{TWA+\\FFM} & \makecell{TWA+\\MFTS} & \makecell{TWA+\\OWA+FFM} & \makecell{TWA+\\FFM} & \makecell{TWA+\\MFTS} \\ 
\midrule
\endhead
\midrule
\multicolumn{8}{r}{{Continued on next page}} \\
\endfoot
\endlastfoot
RMSFE & 1  & 0.0025 & 0.0028 & 0.0031 & 0.0045 & 0.0049 & 0.0054 \\
      & 2  & 0.0028 & 0.0030 & 0.0032 & 0.0044 & 0.0050 & 0.0056 \\
      & 3  & 0.0031 & 0.0031 & 0.0034 & 0.0042 & 0.0050 & 0.0056 \\
      & 4  & 0.0030 & 0.0032 & 0.0035 & 0.0045 & 0.0052 & 0.0058 \\
      & 5  & 0.0026 & 0.0033 & 0.0037 & 0.0053 & 0.0057 & 0.0060 \\
      & 6  & 0.0027 & 0.0035 & 0.0039 & 0.0061 & 0.0063 & 0.0063 \\
      & 7  & 0.0032 & 0.0037 & 0.0042 & 0.0069 & 0.0070 & 0.0066 \\
      & 8  & 0.0040 & 0.0039 & 0.0048 & 0.0063 & 0.0065 & 0.0066 \\
      & 9  & 0.0050 & 0.0044 & 0.0057 & 0.0043 & 0.0054 & 0.0068 \\
      & 10 & 0.0050 & 0.0043 & 0.0057 & 0.0043 & 0.0048 & 0.0063 \\
\cmidrule{2-8} 
      & Mean   & \textbf{0.0034} & 0.0035 & 0.0041 & \textbf{0.0051} & 0.0056 & 0.0061 \\
      & Median & \textbf{0.0030} & 0.0034 & 0.0038 & \textbf{0.0045} & 0.0053 & 0.0061 \\
\midrule
MAFE  & 1  & 0.0010 & 0.0011 & 0.0012 & 0.0017 & 0.0019 & 0.0021 \\
      & 2  & 0.0011 & 0.0011 & 0.0012 & 0.0017 & 0.0019 & 0.0022 \\
      & 3  & 0.0012 & 0.0012 & 0.0013 & 0.0017 & 0.0020 & 0.0023 \\
      & 4  & 0.0012 & 0.0012 & 0.0013 & 0.0019 & 0.0021 & 0.0024 \\
      & 5  & 0.0011 & 0.0013 & 0.0014 & 0.0022 & 0.0023 & 0.0025 \\
      & 6  & 0.0011 & 0.0013 & 0.0014 & 0.0025 & 0.0026 & 0.0026 \\
      & 7  & 0.0012 & 0.0014 & 0.0015 & 0.0028 & 0.0029 & 0.0027 \\
      & 8  & 0.0015 & 0.0014 & 0.0017 & 0.0025 & 0.0026 & 0.0027 \\
      & 9  & 0.0019 & 0.0015 & 0.0019 & 0.0017 & 0.0021 & 0.0026 \\
      & 10 & 0.0019 & 0.0015 & 0.0019 & 0.0016 & 0.0019 & 0.0025 \\
\cmidrule{2-8} 
      & Mean   & \textbf{0.0013} & \textbf{0.0013} & 0.0015 & \textbf{0.0020} & 0.0022 & 0.0025 \\
      & Median & \textbf{0.0012} & 0.0013 & 0.0014 & \textbf{0.0018} & 0.0021 & 0.025 \\
\bottomrule
\end{longtable}
\end{small}
\end{center}

\vspace{-.4in}

For the male series, our developed methods can achieve improved accuracy at the longer horizons, where the number of evaluation points is smaller. For models that can capture the long-term dynamic, it often displays superior accuracy.
To visually inspect the error trajectories across expanding horizons, Figure~\ref{fig:ets_grid} plots the individual RMSFE and MAFE patterns. For the female series, the errors generally grow monotonically with the horizon. Conversely, for the male series, our developed methods achieve notably lower errors at the longest horizons ($h=9, 10$) after peaking around $h=7$, demonstrating their ability to adapt to long-term dynamics where the number of evaluation points is smaller.
\begin{figure}[!htb]
\centering
\begin{tabular}{cc}
\includegraphics[width=0.49\textwidth]{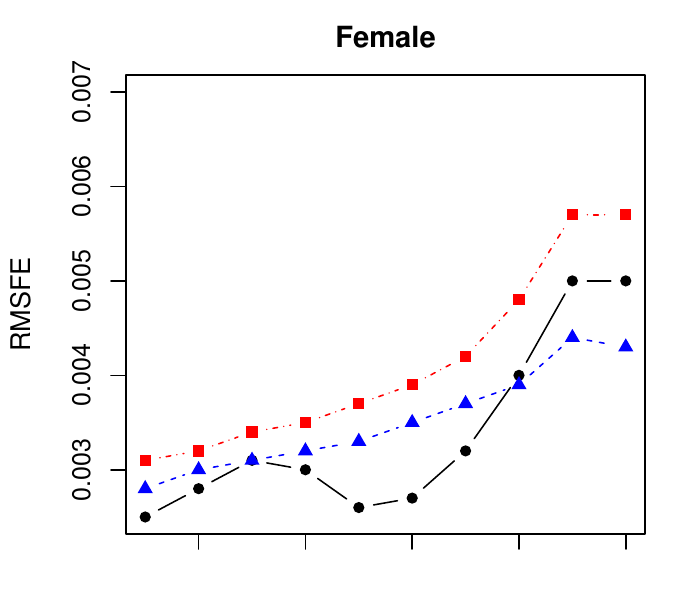} & 
\includegraphics[width=0.49\textwidth]{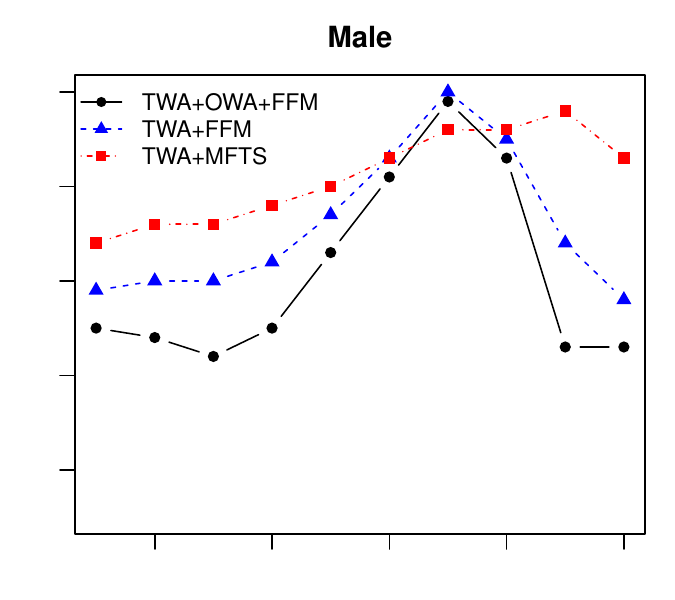} \\ \\[6pt]
\includegraphics[width=0.49\textwidth]{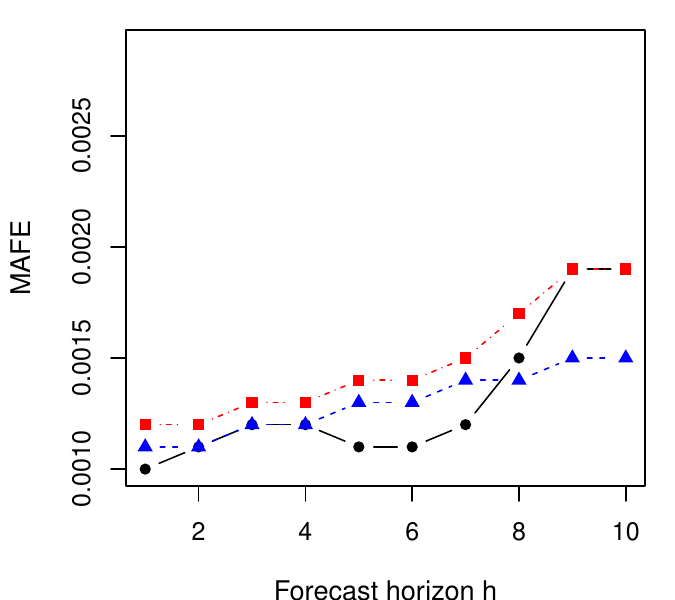} & 
\includegraphics[width=0.49\textwidth]{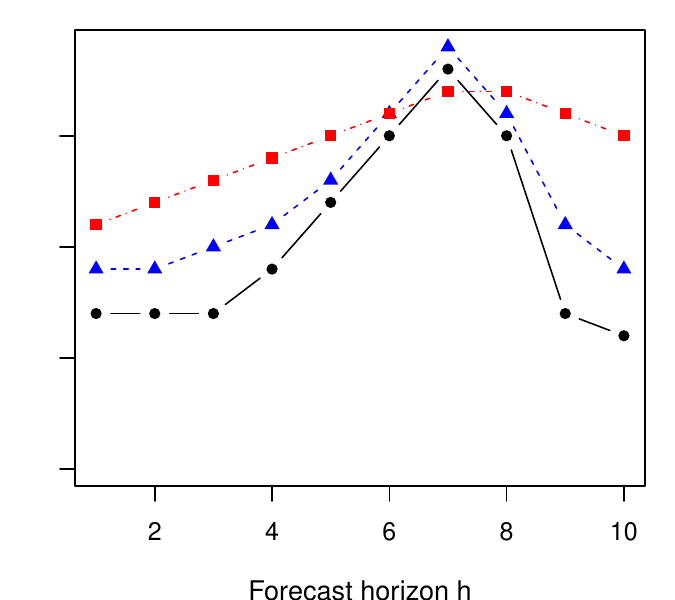} \\ \\
\end{tabular}
\caption{\small ETS point forecast error comparison across expanding horizons $h=1$ to~$10$.}\label{fig:ets_grid}
\end{figure}

In Appendix~\ref{App:A}, we present an additional comparison of point forecast errors based on the ARIMA forecasting method. 

The current results are a gender-level comparison, averaged across 47 prefectures. To examine regional heterogeneity in forecast accuracy, it is possible to average over two genders. To better facilitate this comparison, we develop a \href{https://cfjimenezv.shinyapps.io/FANOVA_FFM/}{Shiny app} in \Rlogo. To formally evaluate the statistical significance of these point forecast accuracy gains across the subnational divisions, we conduct a Model Confidence Set analysis. The detailed empirical framework and the resulting prefecture-level survival heatmaps are provided in Appendix~\ref{App:MCS}.

\subsection{Interval forecast error metrics}\label{sec:IFE}

To evaluate and compare interval forecast accuracy, we compute the empirical coverage probability (ECP), the coverage probability difference (CPD), and the interval score of \cite{GR07}. For each year in the test set, the $h$-step-ahead prediction intervals are calculated at $1-\alpha$ nominal coverage probability, the lower and upper bounds, denoted by $[\widehat{\Z}_{\eta+\xi,i}^{\text{lb},g}(u_j),\widehat{\Z}_{\eta+\xi,i}^{\text{ub},g}(u_j)]$, are not required to be centered around the point forecasts \citep{JSS24, SH26}. For region $i$ and gender $g$, the ECP and CPD are defined as
\begin{align*}
\text{ECP}_{\alpha,h,i}^g &= \frac{1}{(11-h)\times 96}\sum^{10}_{\xi=h}\sum^{96}_{j=1}\mathds{1}\left\{\widehat{\Z}_{\eta+\xi,i}^{\text{lb},g}(u_j)\leq \Z_{\eta+\xi,i}^g(u_j)\leq \widehat{\Z}_{\eta+\xi,i}^{\text{ub},g}(u_j)\right\}\\    
\text{CPD}_{\alpha,h,i}^g &=\left|\frac{1}{(11-h)\times 96}\sum^{10}_{\xi=h}\sum^{96}_{j=1}\left[\mathds{1}\big\{\Z_{\eta+\xi,i}^g(u_j)>\widehat{\Z}_{\eta+\xi,i}^{\text{ub},g}(u_j)\big\}+\mathds{1}\big\{\Z_{\eta+\xi,i}^g(u_j)<\widehat{\Z}_{\eta+\xi,i}^{\text{ub},g}(u_j)\big\}\right]-\alpha\right|,
\end{align*}
where $\eta$ represents the year before the test data set, and $\xi$ denotes the horizon-specific year index of the testing set.

The ECP assesses coverage without evaluating the sharpness of the prediction interval. By combining coverage and sharpness, we use a scoring rule for the prediction interval at age $u_j$, denoted as
\begin{align*}
S_{\alpha,\xi}\left[\widehat{\Z}_{\eta+\xi,i}^{\text{lb}, g}(u_j), \widehat{\Z}_{\eta+\xi,i}^{\text{ub}, g}(u_j), \Z_{\eta+\xi,i}^g(u_j)\right]&=\left[\widehat{\Z}_{\eta+\xi,i}^{\text{ub}, g}(u_j) -\widehat{\Z}_{\eta+\xi,i}^{\text{lb}, g}(u_j)\right] \\
&+\frac{2}{\alpha}\left[\widehat{\Z}_{\eta+\xi,i}^{\text{lb},g}(u_j) - \Z_{\eta+\xi,i}^g(u_j)\right]\mathds{1}\left\{\Z_{\eta+\xi,i}^{g}(u_j)<\widehat{\Z}_{\eta+\xi,i}^{\text{lb},g}(u_j)\right\} \\
&+\frac{2}{\alpha}\left[\Z_{\eta+\xi,i}^g(u_j) - \widehat{\Z}_{\eta+\xi,i}^{\text{ub},g}(u_j)\right]\mathds{1}\left\{\Z_{\eta+\xi,i}^g(u_j)>\widehat{\Z}_{\eta+\xi,i}^{\text{ub},g}(u_j)\right\},
\end{align*}
where $\alpha$ is a level of significance.

Averaging the number of ages and the number of years in the testing set, the mean interval score (IS) is given as
\begin{equation*}
\overline{S}^g_{\alpha,h,i}=\frac{1}{(11-h)\times 96}\sum^{10}_{\xi=h}\sum^{96}_{j=1}S_{\alpha,\xi}\left[\widehat{\Z}_{\eta+\xi,i}^{\text{lb}, g}(u_j), \widehat{\Z}_{\eta+\xi,i}^{\text{ub}, g}(u_j), \Z_{\eta+\xi,i}^g(u_j)\right].
\end{equation*}
Given the same ECP, the mean interval score rewards narrower prediction intervals.

\subsection{Comparison of interval forecast accuracy}\label{sec:CIF}

In Table~\ref{tab:2}, we present the one-step-ahead to ten-step-ahead interval forecast errors based on the split and sequential conformal prediction methods. Between the split and sequential conformal prediction methods, the former leads to under-estimation, in which the empirical coverage probability is smaller than the nominal one. In contrast, the sequential conformal prediction method results in over-estimation, in which the empirical coverage probability is larger than the nominal one. In the split conformal prediction method, the sd is a more accurate summary statistic than the quantile. Between the two univariate time-series forecasting methods, namely the ARIMA and ETS, there is a marginal difference in terms of their interval forecast accuracy.

\begin{footnotesize}
\tabcolsep 0.028in
\renewcommand*{\arraystretch}{1.28}
\begin{longtable}{@{}lc*{17}{c}@{}}
\caption{\small Using the combination of two-way FANOVA, one-way FANOVA and functional factor model, we present the one-step-ahead to ten-step-ahead interval forecast accuracy for the female and male data at the 95\% nominal coverage probability. The methods with the smallest overall CPD and IS are highlighted in bold.\label{tab:2}} \\
\toprule
& \multicolumn{9}{c}{Female} & \multicolumn{9}{c}{Male} \\
\cmidrule(lr){2-10} \cmidrule(lr){11-19}
& \multicolumn{3}{c}{Split (sd)} & \multicolumn{3}{c}{Split (quantile)} & \multicolumn{3}{c}{Sequential} & \multicolumn{3}{c}{Split (sd)} & \multicolumn{3}{c}{Split (quantile)} & \multicolumn{3}{c}{Sequential} \\
\cmidrule(lr){2-4} \cmidrule(lr){5-7} \cmidrule(lr){8-10} \cmidrule(lr){11-13} \cmidrule(lr){14-16} \cmidrule(lr){17-19}
$h$ & ECP & CPD & IS & ECP & CPD & IS & ECP & CPD & IS & ECP & CPD & IS & ECP & CPD & IS & ECP & CPD & IS \\
\midrule
\endfirsthead
\toprule
$h$ & ECP & CPD & IS & ECP & CPD & IS & ECP & CPD & IS & ECP & CPD & IS & ECP & CPD & IS & ECP & CPD & IS \\
\midrule
\endhead
\multicolumn{9}{l}{\hspace{-.07in} \underline{ARIMA}} & \\
1 & 0.942 & 0.017 & 0.008 & 0.908 & 0.046 & 0.007 & 0.971 & 0.031 & 0.013   & 0.942 & 0.016 & 0.013 & 0.912 & 0.042 & 0.012 & 0.980 & 0.033 & 0.021 \\
2 & 0.930 & 0.027 & 0.009 & 0.891 & 0.063 & 0.008 & 0.972 & 0.032 & 0.014 & 0.938 & 0.019 & 0.014 & 0.902 & 0.052 & 0.013 & 0.981 & 0.034 & 0.021 \\
3 & 0.932 & 0.023 & 0.010 & 0.881 & 0.071 & 0.009 & 0.978 & 0.032 & 0.014 & 0.937 & 0.019 & 0.014 & 0.888 & 0.064 & 0.014 & 0.983 & 0.035 & 0.022 \\
4 & 0.930 & 0.025 & 0.011 & 0.871 & 0.081 & 0.009 & 0.981 & 0.035 & 0.016 & 0.933 & 0.023 & 0.015 & 0.878 & 0.075 & 0.015 & 0.987 & 0.037 & 0.023 \\
5 & 0.922 & 0.032 & 0.013 & 0.833 & 0.119 & 0.011 & 0.982 & 0.034 & 0.018 & 0.930 & 0.026 & 0.018 & 0.857 & 0.095 & 0.017 & 0.988 & 0.039 & 0.025 \\
6 & 0.921 & 0.034 & 0.015 & 0.807 & 0.145 & 0.013 & 0.980 & 0.032 & 0.019 & 0.926 & 0.029 & 0.018 & 0.838 & 0.114 & 0.019 & 0.988 & 0.039 & 0.026 \\
7 & 0.922 & 0.032 & 0.018 & 0.779 & 0.173 & 0.014 & 0.975 & 0.029 & 0.020 & 0.922 & 0.034 & 0.021 & 0.815 & 0.137 & 0.022 & 0.990 & 0.040 & 0.027 \\
8 & 0.939 & 0.018 & 0.019 & 0.765 & 0.187 & 0.014 & 0.975 & 0.029 & 0.022 & 0.924 & 0.032 & 0.025 & 0.781 & 0.171 & 0.025 & 0.991 & 0.041 & 0.028 \\
9 & 0.952 & 0.017 & 0.026 & 0.741 & 0.211 & 0.012 & 0.971 & 0.029 & 0.024 & 0.929 & 0.027 & 0.033 & 0.718 & 0.234 & 0.031 & 0.991 & 0.041 & 0.029 \\
10 & 0.940 & 0.024 & 0.085 & 0.673 & 0.277 & 0.012 & 0.969 & 0.028 & 0.025 & 0.935 & 0.023 & 0.094 & 0.618 & 0.332 & 0.043 & 0.993 & 0.043 & 0.030 \\
\cmidrule{1-19}
Mean & 0.933 & \textbf{0.025} & 0.021 & 0.815 & 0.137 & \textbf{0.011} & 0.975 & 0.031 & 0.019 & 0.932 & \textbf{0.025} & 0.027 & 0.821 & 0.132 & \textbf{0.021} & 0.987 & 0.038 & 0.025 \\
Median & 0.931 & \textbf{0.024} & 0.014 & 0.820 & 0.132 & \textbf{0.011} & 0.975 & 0.032 & 0.018 & 0.931 & \textbf{0.025} & 0.018 & 0.848 & 0.104 & \textbf{0.018} & 0.988 & 0.039 & 0.025 \\
\midrule
\multicolumn{9}{l}{\hspace{-.07in} \underline{ETS}}   & \\
1 & 0.927 & 0.028 & 0.008 & 0.882 & 0.072 & 0.007 & 0.963 & 0.030 & 0.013 & 0.940 & 0.017 & 0.013 & 0.910 & 0.044 & 0.012 & 0.979 & 0.033 & 0.020 \\
2 & 0.926 & 0.029 & 0.008 & 0.873 & 0.080 & 0.008 & 0.966 & 0.030 & 0.013 & 0.939 & 0.018 & 0.014 & 0.902 & 0.051 & 0.012 & 0.981 & 0.035 & 0.021 \\
3 & 0.928 & 0.027 & 0.009 & 0.862 & 0.090 & 0.008 & 0.968 & 0.030 & 0.015 & 0.940 & 0.018 & 0.014 & 0.889 & 0.063 & 0.013 & 0.983 & 0.035 & 0.022 \\
4 & 0.934 & 0.024 & 0.009 & 0.855 & 0.097 & 0.008 & 0.971 & 0.033 & 0.016 & 0.940 & 0.018 & 0.015 & 0.882 & 0.070 & 0.013 & 0.986 & 0.037 & 0.023 \\
5 & 0.932 & 0.026 & 0.009 & 0.838 & 0.114 & 0.008 & 0.971 & 0.033 & 0.018 & 0.938 & 0.021 & 0.016 & 0.863 & 0.089 & 0.015 & 0.986 & 0.038 & 0.025 \\
6 & 0.935 & 0.026 & 0.009 & 0.827 & 0.126 & 0.008 & 0.969 & 0.033 & 0.020 & 0.942 & 0.019 & 0.017 & 0.846 & 0.105 & 0.017 & 0.987 & 0.039 & 0.027 \\
7 & 0.941 & 0.022 & 0.012 & 0.815 & 0.137 & 0.009 & 0.968 & 0.032 & 0.020 & 0.943 & 0.018 & 0.021 & 0.831 & 0.121 & 0.019 & 0.985 & 0.037 & 0.026 \\
8 & 0.939 & 0.018 & 0.014 & 0.768 & 0.184 & 0.011 & 0.965 & 0.035 & 0.022 & 0.957 & 0.017 & 0.022 & 0.836 & 0.116 & 0.015 & 0.985 & 0.039 & 0.028 \\
9 & 0.931 & 0.025 & 0.022 & 0.676 & 0.276 & 0.017 & 0.959 & 0.036 & 0.023 & 0.962 & 0.019 & 0.032 & 0.814 & 0.138 & 0.017 & 0.987 & 0.039 & 0.029 \\
10 & 0.918 & 0.040 & 0.056 & 0.555 & 0.395 & 0.026 & 0.955 & 0.038 & 0.024 & 0.956 & 0.018 & 0.124 & 0.742 & 0.208 & 0.022 & 0.988 & 0.040 & 0.030 \\
\cmidrule{1-19}
Mean & 0.931 & \textbf{0.027} & 0.016 & 0.795 & 0.157 & \textbf{0.011} & 0.966 & 0.033 & 0.018 & 0.946 & \textbf{0.018} & 0.029 & 0.851 & 0.101 & \textbf{0.015} & 0.985 & 0.037 & 0.025 \\
Median & 0.931 & \textbf{0.026} & 0.009 & 0.832 & 0.120 & \textbf{0.008} & 0.967 & 0.033 & 0.019 & 0.941 & \textbf{0.018} & 0.017 & 0.854 & 0.097 & \textbf{0.015} & 0.986 & 0.038 & 0.025 \\
\bottomrule
\end{longtable}
\end{footnotesize}

From Table~\ref{tab:2}, it seems that Split sd is consistently competitive, while the Split quantile degrades at longer horizons. The sequential conformal prediction achieves stable high coverage at the cost of wider intervals (larger IS values).

Using the combination of two-way FANOVA and functional factor model, we present the one-step-ahead to ten-step-ahead interval forecast accuracy for the female and male data at the 95\% nominal coverage probability in Appendix~\ref{App:B}. Extending the scope beyond \cite{JSS24}, which considers only $h=1$, we apply our proposed methodology to analyze forecast horizons ranging from $h=1$ to $h=10$. The results of the corresponding prediction interval are reported in Appendix~\ref{App:JCGS_IFE}.

Overall, the split conformal approach calibrated via the standard deviation delivers the most competitive performance among the benchmark methods. Nevertheless, the proposed framework, which combines a FANOVA decomposition with a functional factor model to capture and forecast time-varying residual dynamics, achieves a clear improvement in the accuracy of the interval forecast. In particular, this approach consistently yields smaller interval forecast errors compared to the method proposed in \cite{JSS24}.

\section{Conclusion}\label{sec:7}

We present a novel statistical method for extracting patterns in high-dimensional functional time series and demonstrate its usage using Japanese subnational age- and sex-specific $\log$ mortality rates from 1975 to 2023. Using a two-way FANOVA, we decompose the HDFTS into a functional grand effect, a row effect, a column effect, an interaction term, and time-varying residuals. The time-varying residuals are further modeled via a one-way FANOVA, while the remainder term is captured by a functional factor model. For producing accurate point forecasts, the combination of a two-way FANOVA, a one-way FANOVA, and a functional factor model is recommended. Rather than capturing a generalized prefecture-gender interaction, the one-way FANOVA excels because its prefecture-specific components isolate localized variations. Accounting for these specific effects is precisely what drives the improvement in forecast accuracy. To facilitate reproducibility, the \Rlogo \ code for computing the point and interval forecast errors is available at the \href{https://github.com/cfjimenezv07/FANOVA_FFM}{GitHub repository}.

As a means of quantifying forecast uncertainty, we consider conformal prediction methods. Split conformal prediction requires sample splitting, which can lead to inferior interval forecast accuracy at longer forecast horizons, especially when the summary statistic is the pointwise quantile. From a univariate time-series model, such as AR($p$), the sequential conformal prediction gradually updates the predictive quantiles when new data arrive. Because it does not require calibration using a validation set, this conservative approach with larger IS values is recommended for quantifying finite-sample prediction uncertainty.

There are several ways in which the current paper may be extended, and we briefly outline three:
\begin{inparaenum}[1)]
\item In the sequential conformal prediction, we model the temporal dependence of the absolute residuals via an autoregressive process in a quantile regression. Other time-series models can also be applied.
\item While we consider modeling subnational age- and sex-specific mortality rates, one could also model life-table death counts observed over time, which themselves resemble a time series of probability density functions.
\item We utilize a hierarchical decomposition to incorporate interaction effects, which enhances empirical forecast accuracy. While this paper focuses on the predictive benefits of such terms, developing formal hypothesis tests for examining the significance of the interaction term, accounting for complex spatio-temporal dependencies, presents a valuable direction for future research.
\end{inparaenum}

\section*{Acknowledgment}

The authors are grateful for the insightful comments received from two reviewers and participants at the Recent Advances in Time Series (RATS) workshop in Cyprus. This research is financially supported by the Australian Research Council Discovery Project DP230102250 and the Australian Research Council Future Fellowship FT240100338. Jiménez-Varón gratefully acknowledges support from the EPSRC NeST Program Grant EP/X002195/1.

\newpage
\bibliographystyle{agsm}
\bibliography{FANOVA_FFM.bib}

@Article{SG12,
    author  = {Y. Sun and M. G. Genton},
    journal = {Journal of Agricultural, Biological, and Environmental Statistics},
    title   = {Functional median polish},
    year    = {2012},
    pages   = {354-376},
    volume  = {17}
}

@Article{SME03,
    author  = {D. J. Spitzner and J. S. Marron and G. K. Essick},
    journal = {Journal of the American Statistical Association: Applications \& Case Studies},
    title   = {Mixed-model functional {ANOVA} for studying human tactile perception},
    year    = {2003},
    pages   = {263-272},
    volume  = {98}
}

@article{SK22,
    author = {H. L. Shang and F. Kearney},
    title = {Dynamic functional time-series forecasts of foreign exchange implied volatility surfaces},
    journal = {International Journal of Forecasting},
    year = {2022},
    volume = {38},
    number = {3},
    pages = {1025-1049}
}

@Article{BR98,
    author  = {B. A. Brumback and J. A. Rice},
    journal = {Journal of the American Statistical Association: Theory and Methods},
    title   = {Smoothing spline models for the analysis of nested and crossed samples of curves},
    year    = {1998},
    pages   = {961-976},
    volume  = {93}
}

@Article{WKB03,
    author  = {Y. Wang and C. Ke and M. B. Brown},
    journal = {Biometrics},
    title   = {Shape-invariant modeling of circadian rhythms with random effects and smoothing spline {ANOVA} decomposition},
    year    = {2003},
    pages   = {804-812},
    volume  = {59}
}

@Article{KS10,
    author  = {C. G. Kaufman and S. R. Sain},
    journal = {Bayesian Analysis},
    title   = {Bayesian functional {ANOVA} modeling using {G}aussian process prior distributions},
    year    = {2010},
    pages   = {123-149},
    volume  = {5}
}

@article{ZD23,
    author = {Z. Zhou and H. Dette},
    title = {Statistical inference for high-dimensional panel functional time series},
    journal = {Journal of the Royal Statistical Society: Series B},
    year = {2023},
    volume = {85},
    number = {2},
    pages = {523-549}
}

@article{JSS25,
    author = {C. F. Jim\'{e}nez-Var\'{o}n and Y. Sun and H. L. Shang},
    title = {Forecasting density-valued functional panel data},
    journal = {Australian \& New Zealand Journal of Statistics},
    year = {2025},
    volume = {67},
    number = {3},
    pages = {401-415}
}

@article{CFQ+25,
    author = {J. Chang and Q. Fang and X. Qiao and Q. Yao},
    title = {On the modeling and prediction of high-dimensional functional time series},
    journal = {Journal of the American Statistical Association: Theory and Methods},
    year = {2025},
    volume = {120},
    number = {552},
    pages = {2181-2195}
}

@article{Shang26,
    author = {H. L. Shang},
    title = {Conformal prediction for functional time series: {A}pplication to age-specific mortality rates},
    journal = {Journal of Population Research},
    year = {2026},
    volume = {43},
    pages = {article number 14}
}

@article{MSZ16,
    author = {T. Maiti and S. Sinha and P-S. Zhong},
    title = {Functional mixed effects model for small area estimation},
    journal = {Scandinavian Journal of Statistics},
    year = {2016},
    volume = {43},
    number = {3},
    pages = {886-903}
}

@article{LBK+25,
    author = {T. Loredo and T. Budav\'{a}ri and D. Kent and D. Ruppert},
    title = {Bayesian functional data analysis in astronomy},
    journal = {Physical Sciences Forum},
    year = {2025},
    volume = {12},
    number = {1},
    pages = {12}
}

@Article{LLS24,
    author  = {D. Li and R. Li and H. L. Shang},
    journal = {The Annals of Statistics},
    title   = {Detection and estimation of structural breaks in high-dimensional functional time series},
    year    = {2024},
    number  = {4},
    pages   = {1716-1740},
    volume  = {52}
}

@Article{TSY22,
    author  = {C. Tang and H. L. Shang and Y. Yang},
    journal = {The Annals of Applied Statistics},
    title   = {Clustering and forecasting multiple functional time series},
    year    = {2022},
    pages   = {2523-2553},
    volume  = {16}
}

@article{Wood94,
    author = {Wood, S. N.},
    title = {Monotonic smoothing splines fitted by cross validation},
    journal = {SIAM Journal on Scientific Computing},
    volume = {15},
    number = {5},
    pages = {1126-1133},
    year = {1994},
}

@article{GSY19,
    author = {Y. Gao and H. L. Shang and Y. Yang},
    title = {High-dimensional functional time series forecasting: {A}n application to age-specific mortality rates},
    journal = {Journal of Multivariate Analysis},
    year = {2019},
    volume = {170},
    pages = {232-243}
}

@article{PS23,
    author = {E. Paparoditis and H. L. Shang},
    title = {Bootstrap prediction bands for functional time series},
    journal = {Journal of the American Statistical Association: Theory and Methods},
    year = {2023},
    volume = {118},
    number = {542},
    pages = {972-986}
}

@book{HK12,
    author = {L. Horv\'{a}th and P. Kokoszka},
    title = {Inference for Functional Data with Applications},
    publisher = {Springer},
    year = {2012},
    address = {New York}
}

@article{SH26,
    author = {H. L. Shang and S. Haberman},
    title = {Constructing prediction intervals for the age distribution of deaths},
    journal = {Scandinavian Actuarial Journal},
    year = {2026},
    volume = {2026},
    number = {5},
    pages = {469-486}
}

@Manual{Hyndman25,
    title = {demography: Forecasting Mortality, Fertility, Migration and Population Data},
    author = {Rob Hyndman},
    year = {2025},
    note = {R package version 2.0.1},
    url = {\url{https://CRAN.R-project.org/package=demography}}
}

@article{HS09,
    author = {R. J. Hyndman and H. L. Shang},
    title = {Forecasting functional time series},
    journal = {Journal of the Korean Statistical Society},
    year = {2009},
    volume = {38},
    pages = {199–211}
}

@article{ANH15,
    author = {A. Aue and D. D. Norinho and S. H\"{o}rmann},
    title = {On the prediction of stationary functional time series},
    journal = {Journal of the American Statistical Association: Theory and Methods},
    year = {2015},
    volume = {110},
    number = {509 },
    pages = {378-392}
}

@book{Coulmas07,
    author = {F. Coulmas},
    title = {Population Decline and Ageing in Japan -- the Social Consequences},
    publisher = {Routledge},
    year = {2007},
    address = {New York}
}

@Article{HNT23,
    author  = {M. Hallin and G. Nisol and S. Tavakoli},
    journal = {Journal of Time Series Analysis},
    title   = {Factor models for high-dimensional functional time series {I}: {R}epresentation results},
    year    = {2023},
    pages   = {578-600},
    volume  = {44}
}

@Article{GQW+26,
    author  = {S. Guo and X. Qiao and Q. Wang and Z. Wang},
    journal = {Journal of Business and Economic Statistics},
    title   = {Factor modeling for high-dimensional functional time series},
    year    = {2026},
    number  = {1},
    pages   = {106-119},
    volume  = {44}
}

@Article{TNH23,
    author  = {S. Tavakoli and G. Nisol and M. Hallin},
    journal = {Journal of Time Series Analysis},
    title   = {Factor models for high-dimensional functional time series {II}: {E}stimation and forecasting},
    year    = {2023},
    pages   = {601-621},
    volume  = {44}
}

@Book{Tukey77,
    author    = {J. Tukey},
    publisher = {Addison-Wesley},
    title     = {{Exploratory Data Analysis}},
    year      = {1977},
    address   = {Reading}
}

@book{KR17,
    author = {P. Kokoszka and M. Reimherr},
    title = {Introduction to Functional Data Analysis},
    publisher = {Chapman and Hall/CRC},
    year = {2017},
    address = {New York}
}

@Article{Shang25,
    author  = {H. L. Shang},
    journal = {Journal of Applied Statistics},
    title   = {Forecasting a time series of {L}orenz curves: {O}ne-way functional analysis of variance},
    year    = {2025},
    number  = {15},
    pages   = {2924-2940},
    volume  = {52}
}

@Article{JSS24,
    author  = {C. F. {Jim\'{e}nez-Var\'{o}n} and Y. Sun and H. L. Shang},
    journal = {Journal of Computational and Graphical Statistics},
    title   = {Forecasting high-dimensional functional time series: {A}pplication to sub-national age-specific mortality},
    year    = {2024},
    number  = {4},
    pages   = {1160-1174},
    volume  = {33}
}

@Article{LLS+26,
    author  = {C. Leng and D. Li and H. L. Shang and Y. Xia},
    journal = {Journal of Business and Economic Statistics},
    title   = {Covariance function estimation for high-dimensional functional time series with dual factor structures},
    year    = {2026},
    volume  = {in press}
}

@Manual{JMD26,
    title  = {{National Institute of Population and Social Security Research}},
    author = {{Japanese Mortality Database}},
    note   = {Available at \url{https://www.ipss.go.jp/p-toukei/JMD/index-en.asp} (data downloaded on January 24, 2026)},
    year   = {2026}
}

@Article{HG18,
    author  = {C. Happ and S. Greven},
    journal = {Journal of the American Statistical Association: Theory and Methods},
    title   = {Multivariate functional principal component analysis for data observed on different (dimensional) domains},
    year    = {2018},
    number  = {522},
    pages   = {649-659},
    volume  = {113}
}

@book{RS06,
    title={Functional Data Analysis},
    author={Ramsay, J. and Silverman, B.W.},
    series={Springer Series in Statistics},
    year={2006},
    publisher={Springer},
    edition = {2nd},
    address = {New York}
}

@Book{ZW06,
    author = {E. Zivot and J. Wang},
    title = {{Modeling Financial Time Series with S-PLUS}},
    publisher = {Springer},
    year = {2006},
    address = {New York}
}

@Article{GR07,
    author  = {T. Gneiting and A. E. Raftery},
    journal = {Journal of the American Statistical Association: Review Article},
    title   = {Strictly proper scoring rules, prediction, and estimation},
    year    = {2007},
    number  = {477},
    pages   = {359-378},
    volume  = {102}
}

@ARTICLE{Akaike1969,
    title   = "Power spectrum estimation through autoregressive model fitting",
    author  = "Akaike, Hirotugu",
    journal = "Annals of the Institute of Statistical Mathematics",
    volume  =  21,
    number  =  1,
    pages   = "407--419",
    year    =  1969
}

@article{Rogers2010,
    author = {Rogers, Richard G. and Everett, Bethany G. and Onge, Jarron M. Saint and Krueger, Patrick M.},
    title = {Social, behavioral, and biological factors, and sex differences in mortality},
    journal = {Demography},
    volume = {47},
    number = {3},
    pages = {555-578},
    year = {2010}    
}

@techreport{Liu2026,
    author={Xialu Liu and Xin Wang},
    title={Regularized estimation of the loading matrix in factor models for high-dimensional time series}, 
    institution = {arXiv},
    year = {2026},
    url={\url{https://arxiv.org/abs/2506.11232}}
}

@article{Delaigle2019,
    author = {Delaigle, Aurore and Hall, Peter and Pham, Tung},
    title = {Clustering Functional Data into Groups by Using Projections},
    journal = {Journal of the Royal Statistical Society Series B: Statistical Methodology},
    volume = {81},
    number = {2},
    pages = {271-304},
    year = {2019}
}

@article{MCS,
    author = {Hansen, Peter R. and Lunde, Asger and Nason, James M.},
    title = {The Model Confidence Set},
    journal = {Econometrica},
    volume = {79},
    number = {2},
    pages = {453-497},
    year = {2011}
}

@article{Ikeda11,
    author = {Ikeda, N. and Saito, E. and Kondo, N. and others},
    title = {What has made the population of {Japan} healthy?},
    journal = {The Lancet},
    volume = {378},
    number = {9796},
    pages = {1094--1105},
    year = {2011}
}

@article{Nomura17,
    author = {Nomura, S. and Sakamoto, H. and Glenn, S. and Tsugawa, Y. and Abe, S. K. and Rahman, M. M. and Brown, J. C. and Ezoe, S. and Fitzmaurice, C. and Inokuchi, T. and Kassebaum, N. J. and Kawakami, N. and Kita, Y. and Kondo, N. and Lim, S. S. and Maruyama, S. and Miyata, H. and Mooney, M. D. and Naghavi, M. and Onoda, T. and Ota, E. and Otake, Y. and Roth, G. A. and Saito, E. and Tabuchi, T. and Takasaki, Y. and Tanimura, T. and Uechi, M. and Vos, T. and Wang, H. and Inoue, M. and Murray, C. J. L. and Shibuya, K.},
    title = {Population health and regional variations of disease burden in {Japan}, 1990--2015: {A} systematic subnational analysis for the {Global Burden of Disease Study 2015}},
    journal = {The Lancet},
    volume = {390},
    number = {10101},
    pages = {1521--1538},
    year = {2017}
}

@article{Tsuboi22,
    author = {Tsuboi, S. and Mine, T. and Fukushima, T.},
    title = {Heterogeneous trends of premature mortalities in {Japan}: joinpoint regression analysis of years of life lost from 2011 to 2019},
    journal = {Dialogues in Health},
    volume = {1},
    pages = {100071},
    year = {2022}
}

@article{Willcox07,
    author = {Willcox, B. J. and Willcox, D. C. and Todoriki, H. and Fujiyoshi, A. Yano, K. and He, Q. and Curb, J. D. and Suzuki, M.},
    title = {Caloric restriction, the traditional {Okinawan} diet, and healthy aging: the diet of the world's longest-lived people and its potential impact on morbidity and life span},
    journal = {Annals of the New York Academy of Sciences},
    volume = {1114},
    pages = {434--455},
    year = {2007}
}

\newpage

\appendix

\section{Proof of Proposition~\ref{prop:identifiability}}\label{proof:prop1}

\begin{proof}
The result follows from the fact that the second-stage OWA is applied to the residual process from the first-stage TWA, which is already centered with respect to both prefecture and gender dimensions.

For part {\rm (i)}, consider the estimator of the second-stage functional grand mean:
\begin{equation*}
\widehat{\theta}^{g}(u) = \frac{1}{NT}\sum_{i=1}^{N}\sum_{t=1}^{T}
\X_{it}^{g,\dagger}(u).
\end{equation*}
From the first-stage TWA identifiability constraints in Equation~\eqref{eq:constraints}, the residual process satisfies
\begin{equation*}
\sum_{i=1}^{N}\X_{it}^{g,\dagger}(u)=0, \qquad t=1, 2,\ldots,T.
\end{equation*}
Therefore,
\begin{equation*}
\widehat{\theta}^{g}(u) = \frac{1}{NT}\sum_{t=1}^{T}\left(\sum_{i=1}^{N}\X_{it}^{g,\dagger}(u)\right) = \frac{1}{NT}\sum_{t=1}^{T}0 =0.
\end{equation*}
Hence, the second-stage grand mean is identically zero, implying that it does not provide any additional information and can be omitted from the OWA representation.

For part {\rm (ii)}, the estimated prefecture-specific effect in the second-stage OWA is given by
\begin{equation*}
\widehat{\eta}^{g}_{i}(u)
=
\frac{1}{T}
\sum_{t=1}^{T}
\X_{it}^{g,\dagger}(u)
-
\widehat{\theta}^{g}(u).
\end{equation*}
Since $\widehat{\theta}^{g}(u)\equiv0$ from part {\rm (i)}, this reduces to
\begin{equation*}
\widehat{\eta}^{g}_{i}(u)
=
\frac{1}{T}
\sum_{t=1}^{T}
\X_{it}^{g,\dagger}(u).
\end{equation*}
Summing over prefectures gives
\begin{align*}
\sum_{i=1}^{N}\widehat{\eta}^{g}_{i}(u)
&=
\frac{1}{T}
\sum_{t=1}^{T}
\sum_{i=1}^{N}
\X_{it}^{g,\dagger}(u)\\
&=
\frac{1}{T}
\sum_{t=1}^{T}0
=0,
\end{align*}
where the second equality follows directly from the first-stage residual constraint. Thus, the extracted prefecture-gender interaction profiles satisfy the usual zero-sum identifiability condition.

For part {\rm (iii)}, define the remaining stochastic component after the second-stage decomposition as
\begin{equation*}
\mathcal{R}^{g}_{it}(u)
=
\X_{it}^{g,\dagger}(u)
-
\widehat{\eta}^{g}_{i}(u).
\end{equation*}
Taking the cross-sectional sum at each time point yields
\begin{align*}
\sum_{i=1}^{N}\mathcal{R}^{g}_{it}(u)
&=
\sum_{i=1}^{N}\X_{it}^{g,\dagger}(u)
-
\sum_{i=1}^{N}\widehat{\eta}^{g}_{i}(u).
\end{align*}
The first term is zero by the TWA constraint, while the second term is zero by part {\rm (ii)}. Therefore,
\begin{equation*}
\sum_{i=1}^{N}\mathcal{R}^{g}_{it}(u)=0.
\end{equation*}
Consequently, the sequential TWA--OWA decomposition preserves cross-sectional identifiability at every stage.
\end{proof}

\clearpage
\section{Point forecast accuracy based on ARIMA}\label{App:A}

In Table~\ref{tab:3}, we present an additional comparison of point forecast errors based on the ARIMA forecasting method. Between the ETS and ARIMA forecasting methods, there is an advantage of using the ETS method, as it consistently yields lower mean and median RMSFE and MAFE values across the majority of forecast horizons for both female and male datasets.
\begin{center}
\renewcommand*{\arraystretch}{0.73}
\tabcolsep 0.195in
\begin{small}
\begin{longtable}{@{}llcccccc@{}}
\caption{\small Averaged across 47 prefectures, we evaluate and compare the point forecast accuracy measured by RMSE and MAE. Forecasting method is ARIMA. The method with the smallest overall error is highlighted in bold. TWA + OWA + FFM represents our method combining two-way FANOVA, one-way FANOVA and a functional factor model. TWA + FFM represents our method combining two-way FANOVA and a functional factor model. TWA + MFTS is a benchmark method combining two-way FANOVA and multivariate functional time series method in \cite{JSS24}.}\label{tab:3} \\
\toprule
& & \multicolumn{3}{c}{Female} & \multicolumn{3}{c}{Male} \\
\cmidrule(lr){3-5} \cmidrule(lr){6-8}
Metric & $h$ & \makecell{TWA+\\OWA+FFM} & \makecell{TWA+\\FFM} & \makecell{TWA+\\MFTS} & \makecell{TWA+\\OWA+FFM} & \makecell{TWA+\\FFM} & \makecell{TWA+\\MFTS} \\ 
\midrule
\endfirsthead
\toprule
& & \multicolumn{3}{c}{Female} & \multicolumn{3}{c}{Male} \\
\cmidrule(lr){3-5} \cmidrule(lr){6-8}
Metric & $h$ & \makecell{TWA+\\OWA+FFM} & \makecell{TWA+\\FFM} & \makecell{TWA+\\MFTS} & \makecell{TWA+\\OWA+FFM} & \makecell{TWA+\\FFM} & \makecell{TWA+\\MFTS} \\ 
\midrule
\endhead
\midrule
\multicolumn{8}{r}{{Continued on next page}} \\
\endfoot
\endlastfoot
RMSFE & 1  & 0.0025 & 0.0028 & 0.0034 & 0.0045 & 0.0049 & 0.0055 \\
      & 2  & 0.0028 & 0.0030 & 0.0039 & 0.0044 & 0.0050 & 0.0057 \\
      & 3  & 0.0031 & 0.0031 & 0.0043 & 0.0042 & 0.0050 & 0.0059 \\
      & 4  & 0.0030 & 0.0032 & 0.0048 & 0.0045 & 0.0052 & 0.0061 \\
      & 5  & 0.0026 & 0.0033 & 0.0052 & 0.0053 & 0.0057 & 0.0063 \\
      & 6  & 0.0027 & 0.0035 & 0.0057 & 0.0061 & 0.0063 & 0.0066 \\
      & 7  & 0.0032 & 0.0037 & 0.0063 & 0.0069 & 0.0070 & 0.0071 \\
      & 8  & 0.0040 & 0.0039 & 0.0076 & 0.0063 & 0.0065 & 0.0080 \\
      & 9  & 0.0050 & 0.0044 & 0.0091 & 0.0043 & 0.0054 & 0.0094 \\
      & 10 & 0.0050 & 0.0043 & 0.0095 & 0.0043 & 0.0048 & 0.0094 \\
\cmidrule{2-8} 
      & Mean   & \textbf{0.0034} & 0.0035 & 0.0060 & \textbf{0.0051} & 0.0056 & 0.0070 \\
      & Median & \textbf{0.0030} & 0.0034 & 0.0055 & \textbf{0.0045} & 0.0053 & 0.0064 \\
\midrule
MAFE  & 1  & 0.0010 & 0.0011 & 0.0013 & 0.0017 & 0.0019 & 0.0021 \\
      & 2  & 0.0011 & 0.0011 & 0.0014 & 0.0017 & 0.0019 & 0.0022 \\
      & 3  & 0.0012 & 0.0012 & 0.0015 & 0.0017 & 0.0020 & 0.0023 \\
      & 4  & 0.0012 & 0.0012 & 0.0017 & 0.0019 & 0.0021 & 0.0024 \\
      & 5  & 0.0011 & 0.0013 & 0.0019 & 0.0022 & 0.0023 & 0.0025 \\
      & 6  & 0.0011 & 0.0013 & 0.0021 & 0.0025 & 0.0026 & 0.0026 \\
      & 7  & 0.0012 & 0.0014 & 0.0023 & 0.0028 & 0.0029 & 0.0028 \\
      & 8  & 0.0015 & 0.0014 & 0.0028 & 0.0025 & 0.0026 & 0.0031 \\
      & 9  & 0.0019 & 0.0015 & 0.0033 & 0.0017 & 0.0021 & 0.0036 \\
      & 10 & 0.0019 & 0.0015 & 0.0035 & 0.0016 & 0.0019 & 0.0037 \\
\cmidrule{2-8} 
      & Mean   & \textbf{0.0013} & \textbf{0.0013} & 0.0022 & \textbf{0.0020} & 0.0022 & 0.0027 \\
      & Median & \textbf{0.0012} & 0.0013 & 0.0020 & \textbf{0.0018} & 0.0021 & 0.0025 \\
\bottomrule
\end{longtable}
\end{small}
\end{center}

\vspace{-.6in}

To clearly trace these variations, Figure~\ref{fig:arima_grid} provides the visual counterparts to Table~\ref{tab:3}. The graphical paths show that while the benchmark TWA+MFTS method suffers from a rapid, sharp accumulation of forecast errors as the horizon expands, our proposed functional factor frameworks remain remarkably stable and robust, successfully mitigating error explosion at longer horizons.

\begin{figure}[!htb]
\centering
{\includegraphics[width=0.485\textwidth]{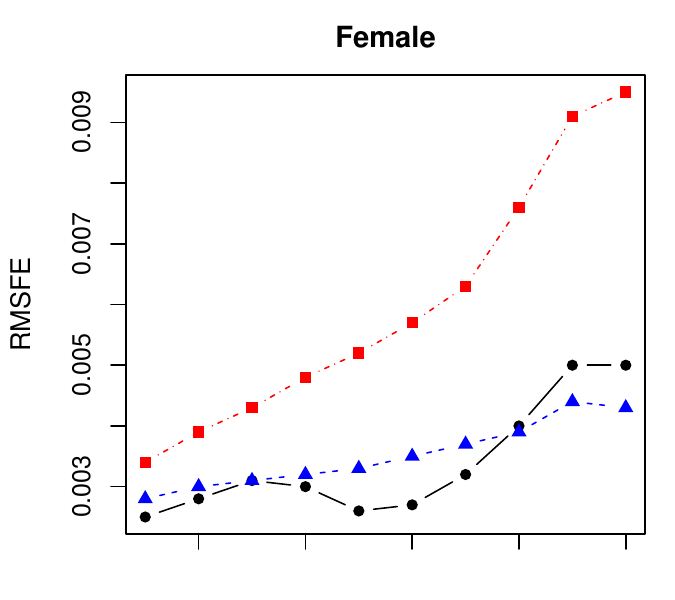}}
\quad
{\includegraphics[width=0.485\textwidth]{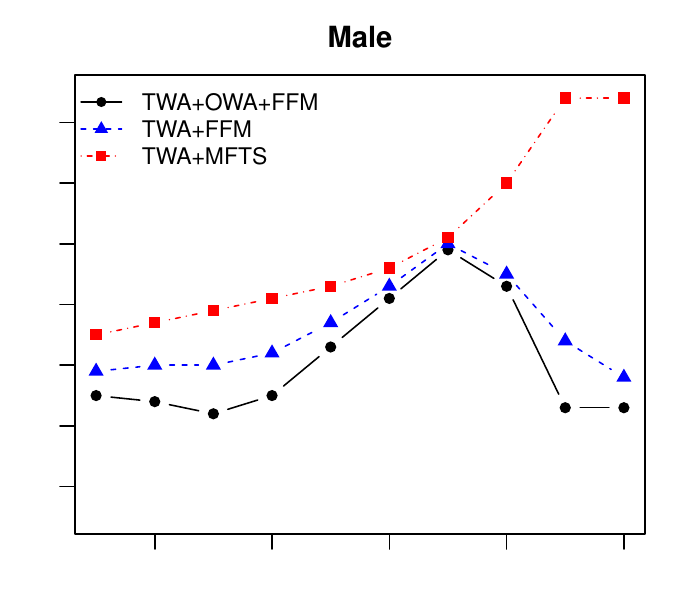}} 
\\
{\includegraphics[width=0.485\textwidth]{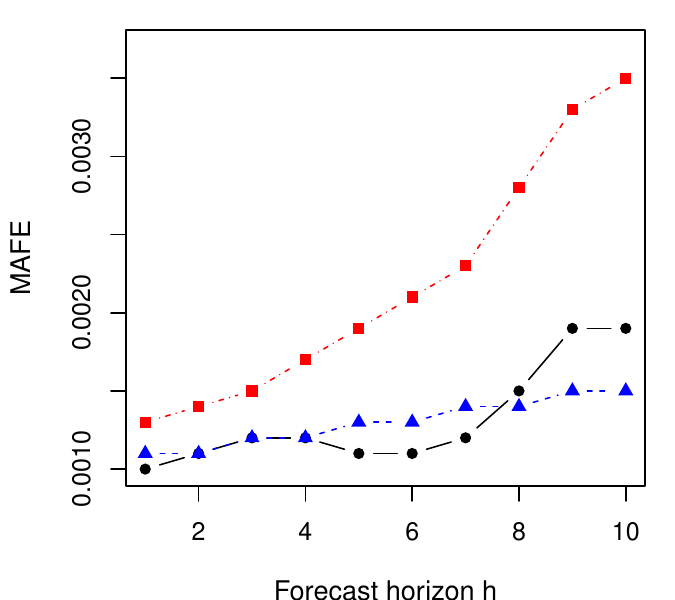}} 
\quad
{\includegraphics[width=0.485\textwidth]{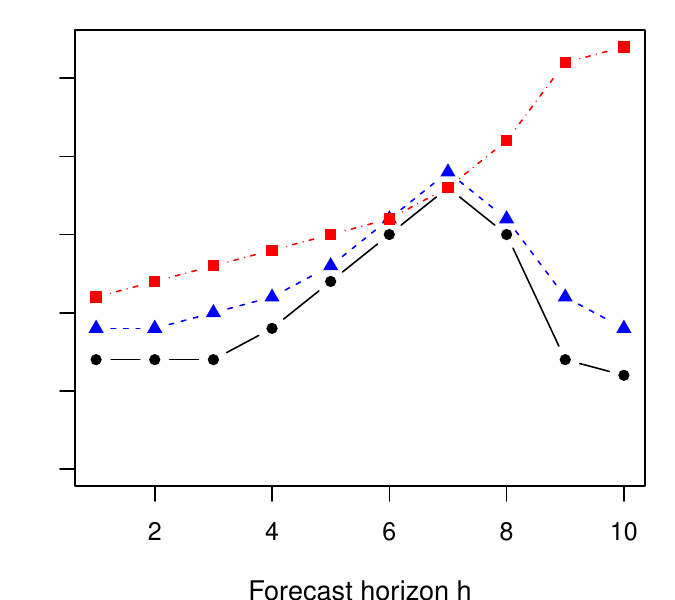}}
\caption{\small ARIMA point forecast error comparison across expanding horizons $h=1$ to $10$.}
\label{fig:arima_grid}
\end{figure}

\clearpage

\section{Interval forecast accuracy from two-way FANOVA and functional factor model}\label{App:B}

Using the combination of two-way FANOVA and functional factor model, in Table~\ref{tab:4}, we present the one-step-ahead to ten-step-ahead interval forecast accuracy for the female and male data at the 95\% nominal coverage probability.
\begin{footnotesize}
\tabcolsep 0.027in
\renewcommand*{\arraystretch}{0.9}
\begin{longtable}{@{}lc*{8}{c}*{9}{c}@{}}
\caption{\small Interval forecast accuracy for the female and male data: one-step-ahead to ten-step-ahead forecasts. TWA+FFM represents two-way FANOVA and functional factor model for the 95\% nominal coverage probability. The methods with the smallest overall CPD and IS are highlighted in bold.}\label{tab:4} \\
\toprule
& \multicolumn{9}{c}{Female} & \multicolumn{9}{c}{Male} \\
\cmidrule(lr){2-10} \cmidrule(lr){11-19}
& \multicolumn{3}{c}{Split(sd)} & \multicolumn{3}{c}{Split(quantile)} & \multicolumn{3}{c}{Seq} & \multicolumn{3}{c}{Split(sd)} & \multicolumn{3}{c}{Split(quantile)} & \multicolumn{3}{c}{Seq} \\
\cmidrule(lr){2-4} \cmidrule(lr){5-7} \cmidrule(lr){8-10} \cmidrule(lr){11-13} \cmidrule(lr){14-16} \cmidrule(lr){17-19}
$h$ & ECP & CPD & IS & ECP & CPD & IS & ECP & CPD & IS & ECP & CPD & IS & ECP & CPD & IS & ECP & CPD & IS \\
\midrule
\endfirsthead
\toprule
$h$ & ECP & CPD & IS & ECP & CPD & IS & ECP & CPD & IS & ECP & CPD & IS & ECP & CPD & IS & ECP & CPD & IS \\
\midrule
\endhead
ARIMA & \\ 
1 & 0.939 & 0.018 & 0.008 & 0.899 & 0.054 & 0.008 & 0.975 & 0.031 & 0.014 & 0.946 & 0.012 & 0.014 & 0.918 & 0.036 & 0.013 & 0.983 & 0.034 & 0.021 \\
2 & 0.932 & 0.023 & 0.009 & 0.888 & 0.066 & 0.009 & 0.976 & 0.032 & 0.014 & 0.938 & 0.019 & 0.014 & 0.905 & 0.049 & 0.014 & 0.984 & 0.034 & 0.021 \\
3 & 0.925 & 0.030 & 0.010 & 0.870 & 0.083 & 0.011 & 0.976 & 0.032 & 0.015 & 0.931 & 0.025 & 0.015 & 0.888 & 0.066 & 0.015 & 0.985 & 0.036 & 0.022 \\
4 & 0.921 & 0.033 & 0.011 & 0.855 & 0.097 & 0.013 & 0.979 & 0.035 & 0.016 & 0.923 & 0.033 & 0.016 & 0.871 & 0.083 & 0.016 & 0.989 & 0.039 & 0.023 \\
5 & 0.912 & 0.041 & 0.013 & 0.833 & 0.119 & 0.016 & 0.979 & 0.034 & 0.018 & 0.912 & 0.044 & 0.018 & 0.849 & 0.105 & 0.019 & 0.990 & 0.040 & 0.024 \\
6 & 0.909 & 0.044 & 0.017 & 0.815 & 0.137 & 0.018 & 0.977 & 0.033 & 0.019 & 0.905 & 0.050 & 0.019 & 0.823 & 0.130 & 0.022 & 0.990 & 0.040 & 0.025 \\
7 & 0.911 & 0.042 & 0.018 & 0.793 & 0.159 & 0.019 & 0.977 & 0.031 & 0.020 & 0.900 & 0.055 & 0.022 & 0.796 & 0.158 & 0.027 & 0.990 & 0.040 & 0.026 \\
8 & 0.933 & 0.022 & 0.018 & 0.787 & 0.165 & 0.016 & 0.975 & 0.030 & 0.021 & 0.898 & 0.057 & 0.025 & 0.764 & 0.189 & 0.030 & 0.991 & 0.041 & 0.027 \\
9 & 0.950 & 0.024 & 0.028 & 0.765 & 0.187 & 0.013 & 0.975 & 0.030 & 0.023 & 0.911 & 0.045 & 0.034 & 0.701 & 0.252 & 0.039 & 0.990 & 0.040 & 0.028 \\
10& 0.908 & 0.048 & 0.088 & 0.695 & 0.255 & 0.012 & 0.975 & 0.029 & 0.024 & 0.896 & 0.057 & 0.092 & 0.610 & 0.340 & 0.048 & 0.992 & 0.042 & 0.029 \\
\cmidrule{1-19}
Mean & 0.924 & \textbf{0.033} & 0.022 & 0.820 & 0.132 & \textbf{0.013} & 0.976 & \textbf{0.032} & 0.018 & 0.916 & 0.040 & 0.027 & 0.812 & 0.141 & 0.024 & 0.988 & \textbf{0.039} & \textbf{0.024} \\
Median & 0.923 & \textbf{0.032} & 0.015 & 0.824 & 0.128 & \textbf{0.013} & 0.976 & \textbf{0.031} & 0.018 & 0.911 & 0.044 & \textbf{0.018} & 0.836 & 0.118 & 0.020 & 0.990 & \textbf{0.040} & 0.025 \\
\midrule
ETS    & \\
1 & 0.938 & 0.018 & 0.008 & 0.898 & 0.055 & 0.007 & 0.973 & 0.030 & 0.013 & 0.947 & 0.012 & 0.014 & 0.918 & 0.036 & 0.013 & 0.981 & 0.032 & 0.020 \\
2 & 0.932 & 0.023 & 0.009 & 0.886 & 0.067 & 0.009 & 0.974 & 0.031 & 0.013 & 0.944 & 0.014 & 0.015 & 0.908 & 0.046 & 0.013 & 0.981 & 0.032 & 0.020 \\
3 & 0.926 & 0.028 & 0.009 & 0.869 & 0.084 & 0.010 & 0.975 & 0.032 & 0.015 & 0.939 & 0.018 & 0.015 & 0.893 & 0.061 & 0.014 & 0.983 & 0.033 & 0.021 \\
4 & 0.923 & 0.032 & 0.010 & 0.857 & 0.096 & 0.010 & 0.977 & 0.033 & 0.015 & 0.932 & 0.023 & 0.016 & 0.873 & 0.081 & 0.016 & 0.984 & 0.035 & 0.021 \\
5 & 0.914 & 0.040 & 0.011 & 0.831 & 0.121 & 0.012 & 0.979 & 0.034 & 0.016 & 0.927 & 0.030 & 0.018 & 0.844 & 0.109 & 0.020 & 0.985 & 0.035 & 0.022 \\
6 & 0.907 & 0.046 & 0.012 & 0.810 & 0.142 & 0.014 & 0.978 & 0.033 & 0.017 & 0.929 & 0.029 & 0.018 & 0.830 & 0.122 & 0.021 & 0.985 & 0.036 & 0.022 \\
7 & 0.906 & 0.048 & 0.013 & 0.785 & 0.167 & 0.015 & 0.977 & 0.032 & 0.017 & 0.940 & 0.022 & 0.020 & 0.832 & 0.121 & 0.020 & 0.984 & 0.035 & 0.022 \\
8 & 0.912 & 0.042 & 0.015 & 0.763 & 0.189 & 0.016 & 0.975 & 0.030 & 0.018 & 0.942 & 0.023 & 0.020 & 0.810 & 0.142 & 0.022 & 0.986 & 0.036 & 0.022 \\
9 & 0.920 & 0.034 & 0.019 & 0.725 & 0.227 & 0.017 & 0.974 & 0.029 & 0.018 & 0.945 & 0.024 & 0.031 & 0.784 & 0.169 & 0.024 & 0.986 & 0.036 & 0.022 \\
10& 0.888 & 0.065 & 0.047 & 0.631 & 0.319 & 0.020 & 0.971 & 0.029 & 0.019 & 0.921 & 0.038 & 0.117 & 0.713 & 0.237 & 0.026 & 0.988 & 0.038 & 0.023 \\
\cmidrule{1-19}
Mean & 0.917 & 0.038 & 0.015 & 0.806 & 0.147 & \textbf{0.013} & 0.975 & \textbf{0.031} & \textbf{0.016} & 0.937 & \textbf{0.023} & 0.028 & 0.841 & 0.112 & \textbf{0.019} & 0.984 & 0.035 & 0.021 \\
Median & 0.917 & 0.037 & \textbf{0.011} & 0.820 & 0.132 & 0.013 & 0.975 & \textbf{0.031} & \textbf{0.016} & 0.940 & \textbf{0.023} & \textbf{0.018} & 0.838 & 0.115 & 0.020 & 0.985 & \textbf{0.035} & 0.022 \\
\bottomrule
\end{longtable}
\end{footnotesize}

\clearpage

\section{Interval forecast accuracy using two-way FANOVA of \texorpdfstring{\cite{JSS24}}{} and conformal prediction}\label{App:JCGS_IFE}

Extending the analysis of \cite{JSS24}, we consider forecast horizons ranging from $h=1$ to $h=10$. The results of the corresponding prediction interval are reported in Table~\ref{tab:5}.
\begin{footnotesize}
\tabcolsep 0.027in
\renewcommand*{\arraystretch}{0.94}
\begin{longtable}{@{} l c *{8}{c}  *{9}{c} @{}}
\caption{\small Interval forecast accuracy for the female and male data: one-step-ahead to ten-step-ahead. TWA+MFTS represents two-way FANOVA and multivariate functional time series method for the 95\% nominal coverage probability. The methods with the smallest overall CPD and IS are highlighted in bold.\label{tab:5}} \\
\toprule
& \multicolumn{9}{c}{Female} & \multicolumn{9}{c}{Male} \\
\cmidrule(lr){2-10} \cmidrule(lr){11-19}
& \multicolumn{3}{c}{Split (sd)} & \multicolumn{3}{c}{Split (quantile)} & \multicolumn{3}{c}{Seq} & \multicolumn{3}{c}{Split (sd)} & \multicolumn{3}{c}{Split (quantile)} & \multicolumn{3}{c}{Seq} \\
\cmidrule(lr){2-4} \cmidrule(lr){5-7} \cmidrule(lr){8-10} \cmidrule(lr){11-13} \cmidrule(lr){14-16} \cmidrule(lr){17-19}
$h$ & ECP & CPD & IS & ECP & CPD & IS & ECP & CPD & IS & ECP & CPD & IS & ECP & CPD & IS & ECP & CPD & IS \\
\midrule
\endfirsthead
\toprule
$h$ & ECP & CPD & IS & ECP & CPD & IS & ECP & CPD & IS & ECP & CPD & IS & ECP & CPD & IS & ECP & CPD & IS \\
\midrule
\endhead
ARIMA & \\ 
1  & 0.930 & 0.028 & 0.010 & 0.880 & 0.075 & 0.010 & 0.982 & 0.036 & 0.024 & 0.938 & 0.019 & 0.015 & 0.912 & 0.043 & 0.014 & 0.990 & 0.040 & 0.028 \\
2  & 0.919 & 0.037 & 0.026 & 0.860 & 0.095 & 0.013 & 0.981 & 0.035 & 0.023 & 0.932 & 0.023 & 0.016 & 0.901 & 0.053 & 0.015 & 0.990 & 0.040 & 0.027 \\
3  & 0.916 & 0.039 & 0.011 & 0.842 & 0.110 & 0.012 & 0.981 & 0.035 & 0.023 & 0.927 & 0.025 & 0.016 & 0.884 & 0.066 & 0.016 & 0.989 & 0.041 & 0.027 \\
4  & 0.913 & 0.041 & 0.013 & 0.824 & 0.128 & 0.013 & 0.982 & 0.035 & 0.024 & 0.920 & 0.033 & 0.016 & 0.870 & 0.081 & 0.017 & 0.990 & 0.041 & 0.028 \\
5  & 0.903 & 0.050 & 0.016 & 0.800 & 0.152 & 0.015 & 0.983 & 0.036 & 0.022 & 0.912 & 0.040 & 0.017 & 0.850 & 0.101 & 0.018 & 0.991 & 0.042 & 0.029 \\
6  & 0.904 & 0.051 & 0.016 & 0.780 & 0.173 & 0.017 & 0.982 & 0.035 & 0.021 & 0.910 & 0.043 & 0.019 & 0.824 & 0.127 & 0.020 & 0.989 & 0.041 & 0.027 \\
7  & 0.904 & 0.051 & 0.019 & 0.750 & 0.203 & 0.020 & 0.981 & 0.032 & 0.021 & 0.909 & 0.042 & 0.021 & 0.800 & 0.150 & 0.023 & 0.988 & 0.042 & 0.028 \\
8  & 0.915 & 0.041 & 0.024 & 0.717 & 0.236 & 0.023 & 0.979 & 0.031 & 0.021 & 0.915 & 0.038 & 0.024 & 0.768 & 0.182 & 0.026 & 0.988 & 0.042 & 0.028 \\
9  & 0.931 & 0.029 & 0.033 & 0.658 & 0.294 & 0.026 & 0.977 & 0.030 & 0.020 & 0.917 & 0.039 & 0.035 & 0.710 & 0.241 & 0.033 & 0.989 & 0.043 & 0.027 \\
10 & 0.913 & 0.041 & 0.075 & 0.564 & 0.386 & 0.032 & 0.976 & 0.028 & 0.020 & 0.916 & 0.038 & 0.089 & 0.611 & 0.339 & 0.045 & 0.987 & 0.043 & 0.027 \\
\cmidrule{1-19}
Mean   & 0.915 & 0.041 & 0.024 & 0.767 & 0.185 & \textbf{0.018} & 0.981 & \textbf{0.033} & 0.022 & 0.920 & \textbf{0.034} & 0.027 & 0.813 & 0.138 & \textbf{0.023} & 0.989 & 0.041 & 0.028 \\
Median & 0.914 & 0.041 & 0.018 & 0.790 & 0.162 & \textbf{0.016} & 0.981 & \textbf{0.035} & 0.022 & 0.917 & \textbf{0.038} & \textbf{0.018} & 0.837 & 0.114 & \textbf{0.019} & 0.989 & 0.042 & 0.028 \\
\midrule
ETS    & \\
1  & 0.935 & 0.022 & 0.008 & 0.889 & 0.065 & 0.008 & 0.980 & 0.034 & 0.023 & 0.941 & 0.016 & 0.015 & 0.911 & 0.044 & 0.014 & 0.989 & 0.040 & 0.028 \\
2  & 0.929 & 0.027 & 0.009 & 0.877 & 0.077 & 0.008 & 0.981 & 0.033 & 0.023 & 0.938 & 0.018 & 0.016 & 0.899 & 0.055 & 0.015 & 0.989 & 0.039 & 0.027 \\
3  & 0.923 & 0.032 & 0.010 & 0.860 & 0.092 & 0.009 & 0.981 & 0.035 & 0.023 & 0.931 & 0.024 & 0.016 & 0.883 & 0.067 & 0.016 & 0.989 & 0.039 & 0.027 \\
4  & 0.918 & 0.036 & 0.010 & 0.848 & 0.104 & 0.010 & 0.981 & 0.035 & 0.024 & 0.928 & 0.026 & 0.017 & 0.871 & 0.079 & 0.017 & 0.989 & 0.039 & 0.028 \\
5  & 0.915 & 0.040 & 0.011 & 0.828 & 0.124 & 0.011 & 0.983 & 0.036 & 0.023 & 0.928 & 0.027 & 0.018 & 0.855 & 0.095 & 0.019 & 0.989 & 0.039 & 0.028 \\
6  & 0.912 & 0.043 & 0.012 & 0.812 & 0.140 & 0.012 & 0.981 & 0.035 & 0.022 & 0.925 & 0.031 & 0.019 & 0.836 & 0.114 & 0.021 & 0.988 & 0.038 & 0.027 \\
7  & 0.909 & 0.046 & 0.013 & 0.788 & 0.164 & 0.013 & 0.978 & 0.033 & 0.022 & 0.928 & 0.029 & 0.021 & 0.823 & 0.127 & 0.022 & 0.987 & 0.038 & 0.027 \\
8  & 0.911 & 0.042 & 0.016 & 0.757 & 0.195 & 0.016 & 0.978 & 0.032 & 0.021 & 0.935 & 0.028 & 0.022 & 0.805 & 0.145 & 0.023 & 0.986 & 0.037 & 0.027 \\
9  & 0.918 & 0.037 & 0.022 & 0.711 & 0.240 & 0.020 & 0.976 & 0.031 & 0.021 & 0.940 & 0.025 & 0.034 & 0.776 & 0.174 & 0.026 & 0.988 & 0.038 & 0.027 \\
10 & 0.898 & 0.054 & 0.049 & 0.607 & 0.343 & 0.029 & 0.973 & 0.032 & 0.020 & 0.937 & 0.030 & 0.101 & 0.704 & 0.246 & 0.029 & 0.989 & 0.038 & 0.027 \\
\cmidrule{1-19}
Mean   & 0.917 & 0.038 & 0.016 & 0.798 & 0.154 & \textbf{0.014} & 0.979 & \textbf{0.034} & 0.022 & 0.933 & \textbf{0.025} & 0.028 & 0.836 & 0.115 & \textbf{0.020} & 0.988 & 0.038 & 0.027 \\
Median & 0.916 & 0.039 & \textbf{0.011} & 0.820 & 0.132 & \textbf{0.011} & 0.981 & \textbf{0.034} & 0.022 & 0.933 & 0.026 & \textbf{0.019} & 0.846 & 0.105 & \textbf{0.020} & 0.989 & 0.039 & 0.027 \\
\bottomrule
\end{longtable}
\end{footnotesize}

\clearpage

\vspace{-0.7in}
\section{Model Confidence Set (MCS) Analysis}\label{App:MCS}

We conduct a Model Confidence Set (MCS) analysis following \cite{MCS} to examine whether the point forecast accuracy improvements achieved by our proposed framework are statistically significant. The MCS procedure addresses this by isolating a set of superior models, where the null hypothesis of equal predictive ability cannot be rejected at a significance level of $\alpha = 0.10$.

The evaluation uses Mean Absolute Forecast Error (MAFE) and Root Mean Squared Forecast Error (RMSFE) panels over 10 forecast horizons ($h=1$ to $h=10$). The loss matrix pools all 10 expanding windows and 47 Japanese prefectures simultaneously. Consequently, the resulting heatmaps reflect framework survival across the entire regional panel; a score of 47 represents uniform inclusion in the superior set across all prefectures, while 0 represents absolute statistical elimination.

\subsection{Empirical Results and Analysis}

The MCS survival counts across all forecast horizons are displayed in the heatmaps below. Specifically, Figure~\ref{fig:mcs1} presents the results for the female cohort, while Figure~\ref{fig:mcs2} illustrates the corresponding evaluations for the male cohort.

\subsubsection{ARIMA Base Model Dominance}

Under ARIMA base specifications, the proposed TWA+OWA+FFM framework displays absolute dominance. As shown in the top panels of Figure~\ref{fig:mcs1} and Figure~\ref{fig:mcs2}, across both genders and metrics, TWA+OWA+FFM is the sole survivor with a perfect score of $47$ from horizons $h=1$ to $h=7$, while TWA+MFTS and TWA+FFM are rejected entirely with a score of 0. In the long-run horizons ($h=8$ to $h=10$), TWA+FFM co-survives with a score of $47$ alongside our model for females (Figure~\ref{fig:mcs1}), but TWA+OWA+FFM remains the unique winner with a score of $47$ for males across all $10$ horizons (Figure~\ref{fig:mcs2}).

\begin{figure}[!htb]
\centering
\begin{tabular}{cc}
\includegraphics[width=0.48\textwidth]{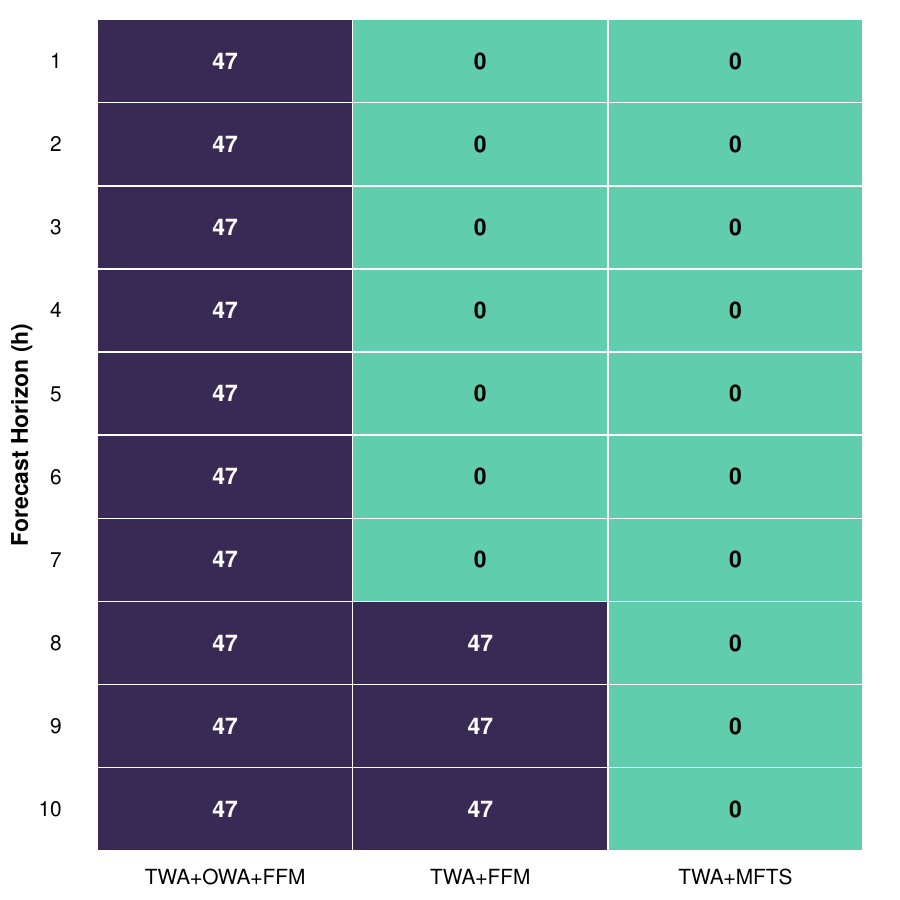} & 
\includegraphics[width=0.48\textwidth]{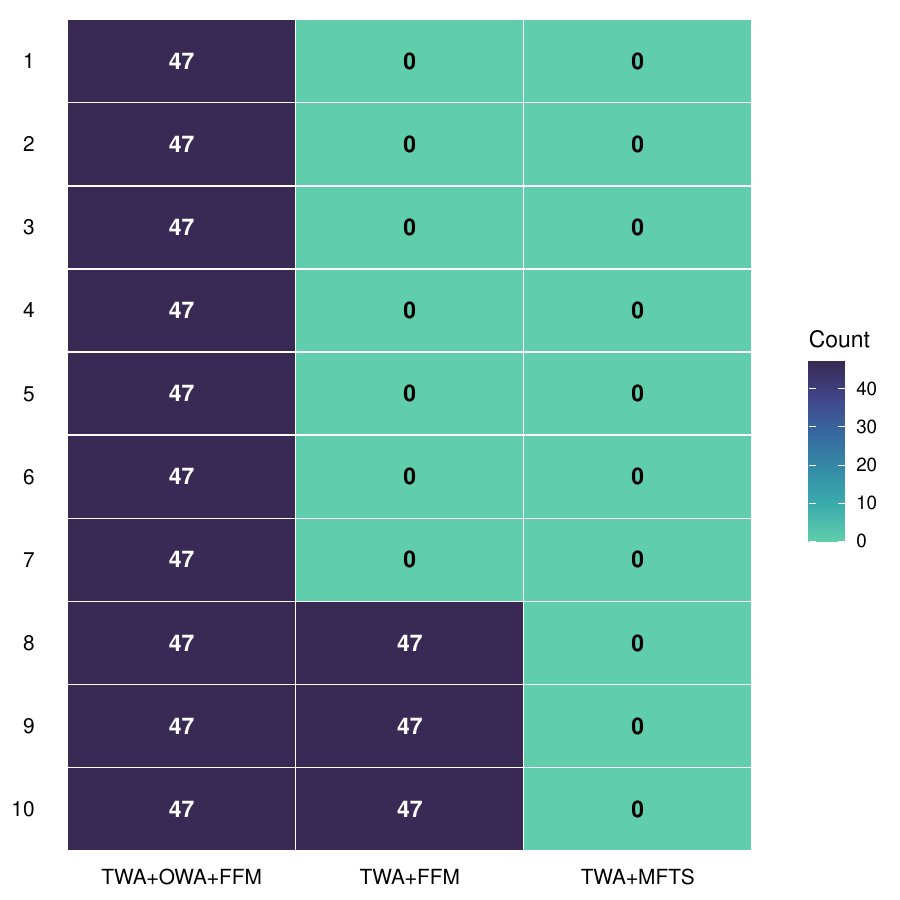} \\
(a) ARIMA -- MAFE Evaluation & (b) ARIMA -- RMSFE Evaluation \\[6pt]
\includegraphics[width=0.48\textwidth]{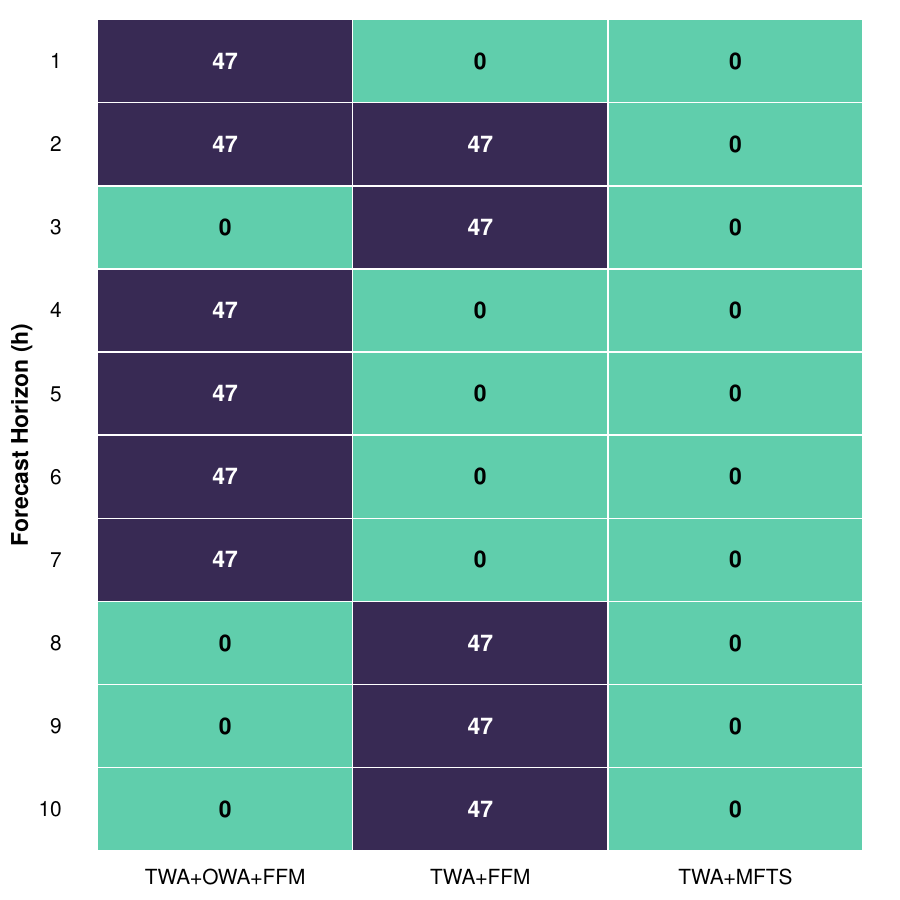} & 
\includegraphics[width=0.48\textwidth]{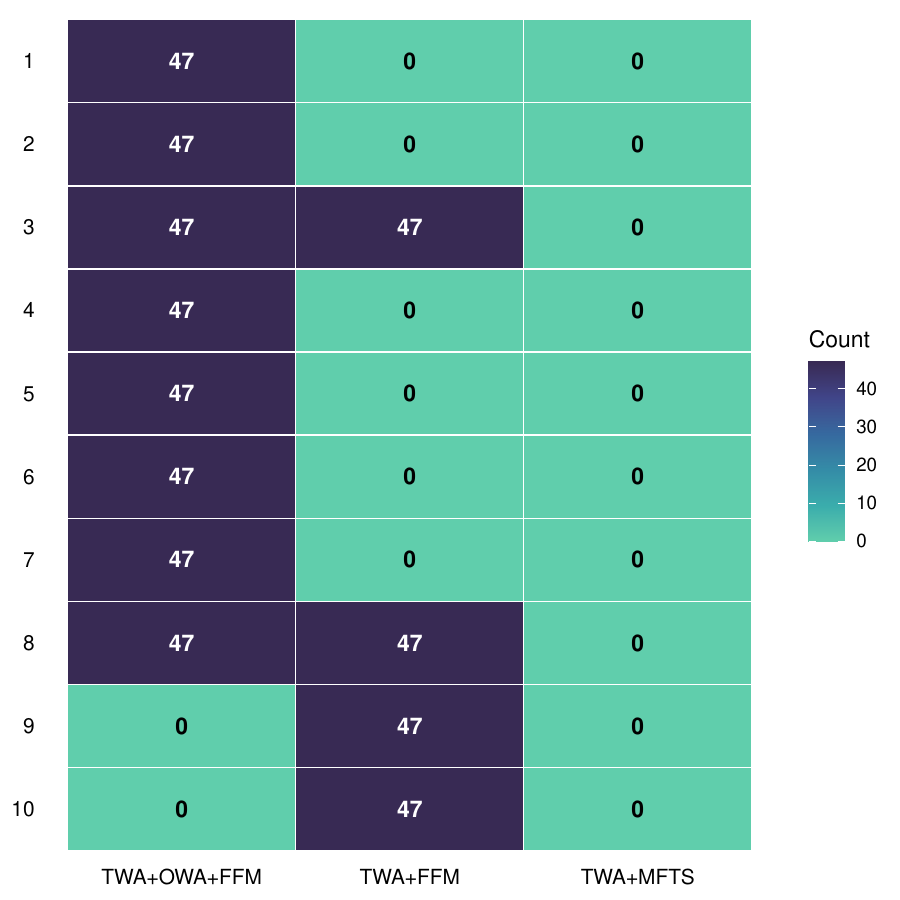} \\
(c) ETS -- MAFE Evaluation & (d) ETS -- RMSFE Evaluation \\
\end{tabular}
\caption{\small Model confidence set (MCS) survival counts for female cohort (significance level $\alpha = 0.10$, implying the 90\% confidence level)}\label{fig:mcs1}
\end{figure}

\subsubsection{ETS Base Model Evaluation}

The ETS specifications, detailed in the bottom panels of Figure~\ref{fig:mcs1} and Figure~\ref{fig:mcs2}, confirm standalone superiority for TWA+OWA+FFM in the short run from horizons $h=1$ to $h=5$. For males (Figure~\ref{fig:mcs2}), localized performance overlaps emerge at horizons $h=6$ to $h=8$ where all three frameworks co-survive with scores of $47$, before TWA+OWA+FFM regains standalone dominance at horizon $h=9$. For females (Figure~\ref{fig:mcs1}), TWA+FFM displays strong longer-horizon performance, matching our model at horizon $h=2$ and $h=8$, and becoming the sole survivor at horizons $h=9$ and $h=10$. The alternative baseline TWA+MFTS is eliminated across almost all horizons for both genders, except during the shared mid-horizon ties observed for males.

\begin{figure}[!htb]
\centering
\begin{tabular}{cc}
\includegraphics[width=0.48\textwidth]{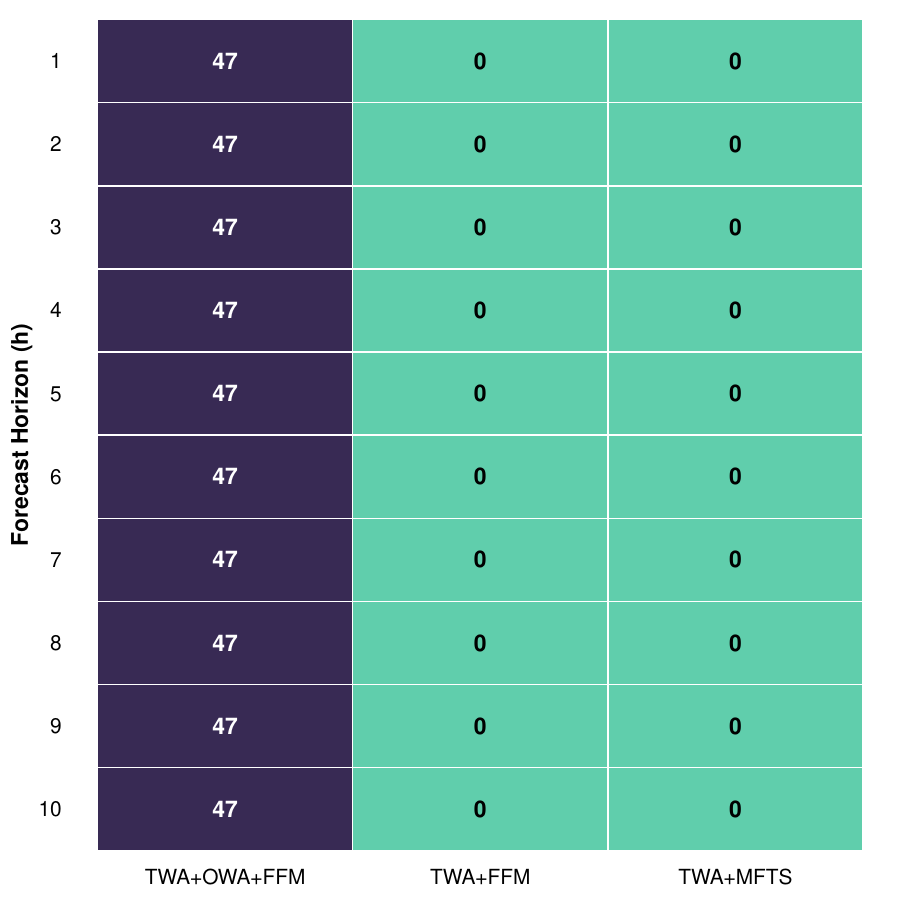} & 
\includegraphics[width=0.48\textwidth]{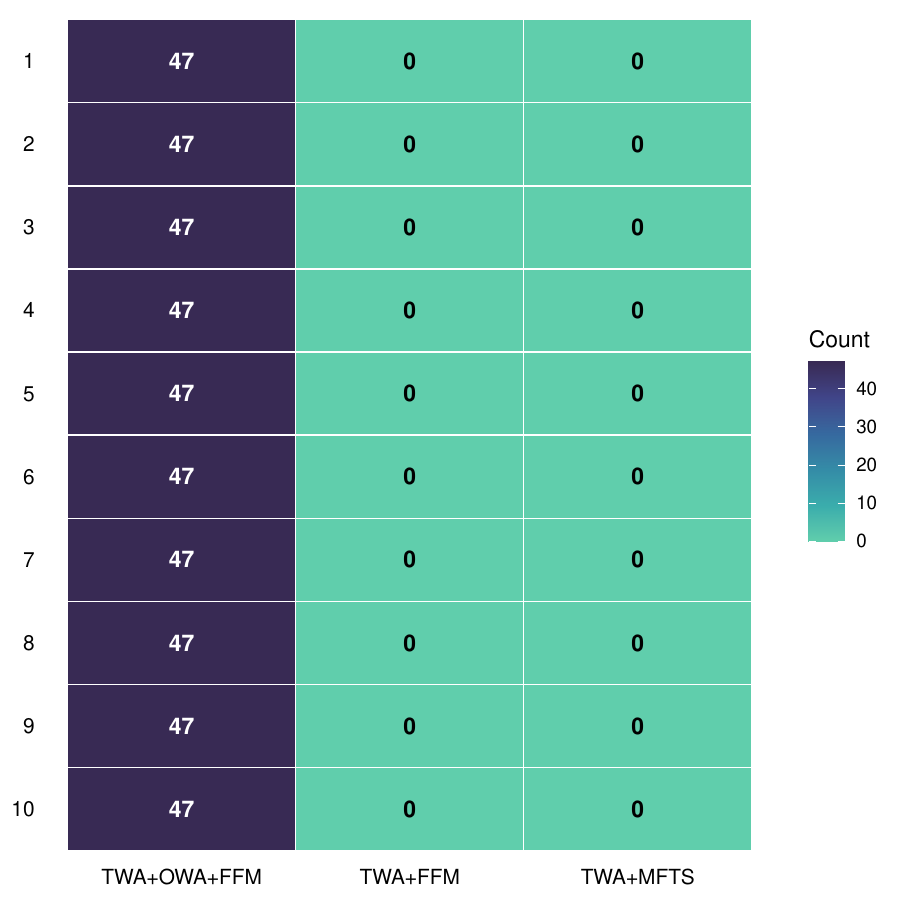} \\
(a) ARIMA -- MAFE Evaluation & (b) ARIMA -- RMSFE Evaluation \\[6pt]
\includegraphics[width=0.48\textwidth]{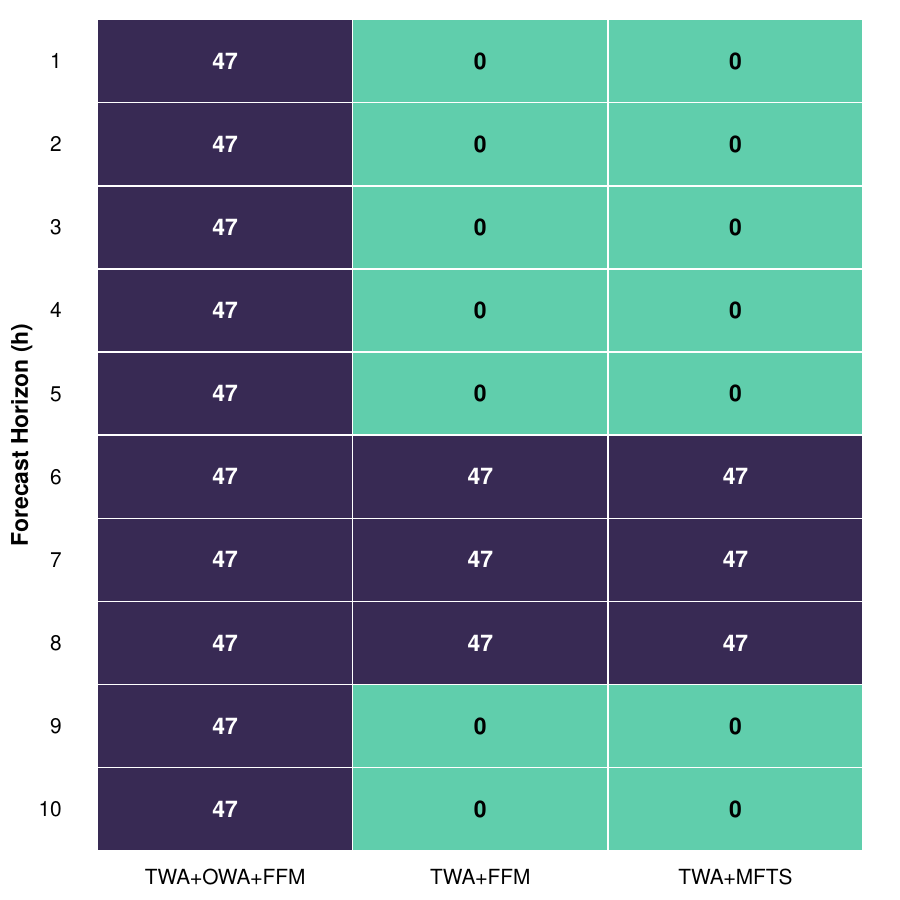} & 
\includegraphics[width=0.48\textwidth]{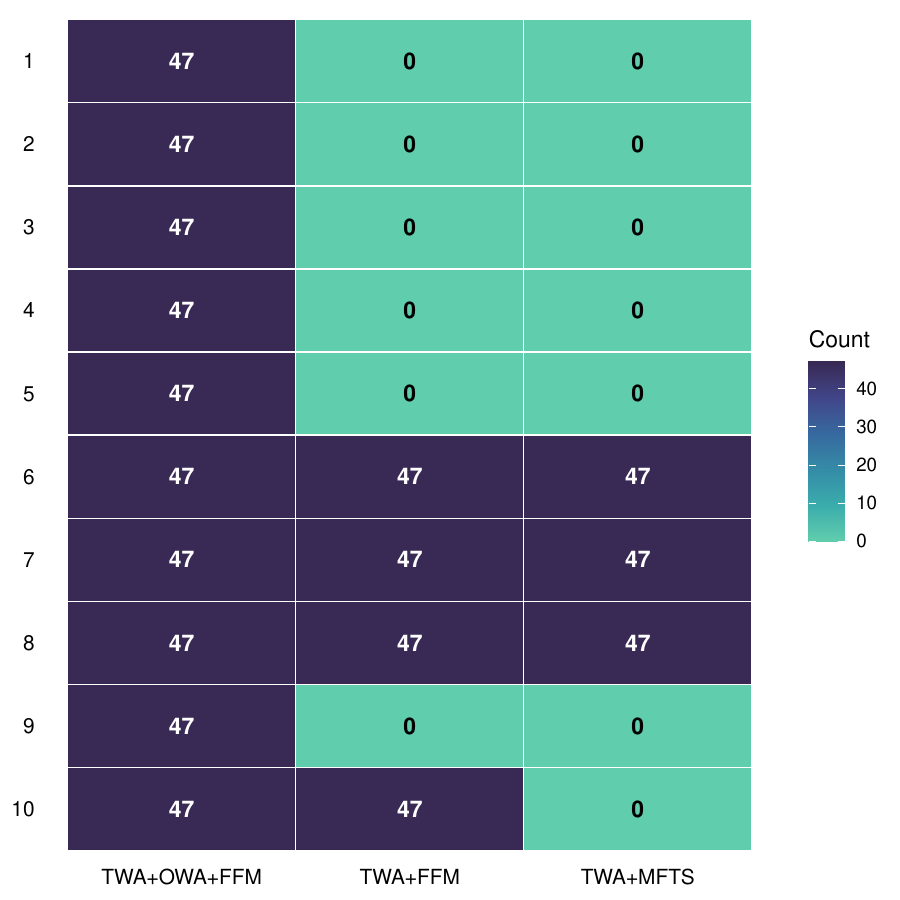} \\
(c) ETS -- MAFE Evaluation & (d) ETS -- RMSFE Evaluation \\
\end{tabular}
\caption{\small Model confidence set (MCS) survival counts for male cohort (significance level $\alpha = 0.10$, implying the 90\% confidence level).}\label{fig:mcs2}
\end{figure}
\clearpage

\vspace{-.67in} 

\section{Automated Model Selection Results}
\label{app:model_specs}

This appendix provides empirical evidence regarding the structural choices made by the automated selection procedures. Specifically, we document the estimated number of functional factors determined by the factor decomposition framework outlined in Section~\ref{sec:3.3}, alongside the univariate parameter specifications chosen for the ARIMA and ETS models across the 10 expanding forecast windows.

\subsection{Factor Selection Results ($\widehat{q}$)}

To evaluate the stability of the model dimension over time, the optimal number of functional factors~($\widehat{q}$) was determined independently at each forecast origin. Following the HDFTS framework, $\widehat{q}$ is selected via an ensemble routine that takes the maximum recommendation of three eigenvalue-based metrics, namely a $k$-criterion measuring consecutive eigenvalue differences, a standard threshold criterion~($\tau = 10^{-3}$), and a novel penalization criterion balancing sample and population dimensions, bounded at an upper limit of~6.

As shown in Table~\ref{tab:factor_stability}, this allocation routine yields a highly consistent factor dimension: the optimal number of factors is stable at $\widehat{q} = 1$ for both female and male mortality series across all 10 horizons under both the proposed TWA-OWA-FFM and the benchmark TWA-FFM frameworks. For the subnational TWA-MFTS framework across the 47 prefectures, $\widehat{q} = 1$ is preferred across all windows for both sexes, with minor exceptions (such as Iwate and Miyagi) where a stable dimension of $\widehat{q} = 2$ is selected for the female series.
\begin{table}[!htb]
\centering
\caption{\small Selected number of factors ($\widehat{q}$) across expanding windows}
\label{tab:factor_stability}
\resizebox{\textwidth}{!}{%
\begin{tabular}{@{}lcccccccccc@{}}
\toprule
\textbf{Framework / Cohort} & \textbf{1} & \textbf{2} & \textbf{3} & \textbf{4} & \textbf{5} & \textbf{6} & \textbf{7} & \textbf{8} & \textbf{9} & \textbf{10} \\ \hline
\textit{TWA-OWA-FFM \& TWA-FFM} & & & & & & & & & & \\
Female $\widehat{q}$ & 1 & 1 & 1 & 1 & 1 & 1 & 1 & 1 & 1 & 1 \\
Male $\widehat{q}$   & 1 & 1 & 1 & 1 & 1 & 1 & 1 & 1 & 1 & 1 \\ 
\midrule
\textit{TWA-MFTS (Selected Prefectures)} & & & & & & & & & & \\
Hokkaido (Male / Female) & 1 / 1 & 1 / 1 & 1 / 1 & 1 / 1 & 1 / 1 & 1 / 1 & 1 / 1 & 1 / 1 & 1 / 1 & 1 / 1 \\
Aomori (Male / Female)   & 1 / 1 & 1 / 1 & 1 / 1 & 1 / 1 & 1 / 1 & 1 / 1 & 1 / 1 & 1 / 1 & 1 / 1 & 1 / 1 \\
Iwate (Male / Female)    & 1 / 2 & 1 / 2 & 1 / 2 & 1 / 2 & 1 / 2 & 1 / 2 & 1 / 2 & 1 / 2 & 1 / 2 & 1 / 2 \\
Miyagi (Male / Female)   & 1 / 2 & 1 / 2 & 1 / 2 & 1 / 2 & 1 / 2 & 1 / 2 & 1 / 2 & 1 / 2 & 1 / 2 & 1 / 2 \\ 
\bottomrule
\end{tabular}%
}
\end{table}

\subsection{ARIMA and ETS Specifications}

Table~\ref{tab:model_specifications} details the specifications selected by the automated univariate forecasting routines in \Rlogo \ (\texttt{auto.arima} and \texttt{ets}). 

\begin{table}[!htb]
\centering
\caption{\small Automated ARIMA and ETS specifications by estimation window}
\label{tab:model_specifications}
\small
\begin{tabular*}{\textwidth}{@{\extracolsep{\fill}}lllllc@{\extracolsep{\fill}}}
\toprule
\textbf{Window} & \textbf{Sex} & \textbf{ARIMA Specification} & \textbf{ETS Specification} & \textbf{Damping} & \textbf{Factor Rank} \\ \midrule
1 & Male   & ARIMA(0,1,1) with drift & ETS(A,Ad,N) & Damped     & 1 \\
1 & Female & ARIMA(0,2,2)            & ETS(A,Ad,N) & Damped     & 1 \\ 
\\
2 & Male   & ARIMA(0,1,1) with drift & ETS(A,Ad,N) & Damped     & 1 \\
2 & Female & ARIMA(0,2,2)            & ETS(A,Ad,N) & Damped     & 1 \\ 
\\
3 & Male   & ARIMA(0,1,1) with drift & ETS(A,Ad,N) & Damped     & 1 \\
3 & Female & ARIMA(0,2,2)            & ETS(A,Ad,N) & Damped     & 1 \\ 
\\
4 & Male   & ARIMA(0,1,1) with drift & ETS(A,Ad,N) & Damped     & 1 \\
4 & Female & ARIMA(0,2,2)            & ETS(A,Ad,N) & Damped     & 1 \\ 
\\
5 & Male   & ARIMA(0,1,1) with drift & ETS(A,A,N)  & Not Damped & 1 \\
5 & Female & ARIMA(0,2,2)            & ETS(A,Ad,N) & Damped     & 1 \\ 
\\
6 & Male   & ARIMA(0,1,1) with drift & ETS(A,A,N)  & Not Damped & 1 \\
6 & Female & ARIMA(0,2,2)            & ETS(A,Ad,N) & Damped     & 1 \\ 
\\
7 & Male   & ARIMA(0,1,1) with drift & ETS(A,Ad,N) & Damped     & 1 \\
7 & Female & ARIMA(0,2,2)            & ETS(A,Ad,N) & Damped     & 1 \\ 
\\
8 & Male   & ARIMA(0,1,1) with drift & ETS(A,Ad,N) & Damped     & 1 \\
8 & Female & ARIMA(0,2,2)            & ETS(A,Ad,N) & Damped     & 1 \\ 
\\
9 & Male   & ARIMA(0,1,1) with drift & ETS(A,Ad,N) & Damped     & 1 \\
9 & Female & ARIMA(0,2,2)            & ETS(A,Ad,N) & Damped     & 1 \\ 
\\
10 & Male  & ARIMA(0,2,2)            & ETS(A,Ad,N) & Damped     & 1 \\
10 & Female& ARIMA(0,2,2)            & ETS(A,Ad,N) & Damped     & 1 \\ \bottomrule
\end{tabular*}
\end{table}

The parameter selections reveal substantial consistency across time:
\begin{itemize}
\item \textbf{ARIMA Specifications:} Under both proposed frameworks (TWA-OWA-FFM and TWA-FFM), the female factor is systematically modeled as an $\text{ARIMA}(0,2,2)$ process. The male factor follows an $\text{ARIMA}(0,1,1)$ with drift model for windows 1 through 9, shifting to an $\text{ARIMA}(0,2,2)$ specification at the final forecast origin. For the subnational TWA-MFTS framework of \cite{JSS24}, an $\text{ARIMA}(0,1,1)$ with drift process is selected for 55.5\% (533 out of 960) of the total paths.
\item \textbf{ETS Specifications and Damping:} The automated ETS framework is represented using the standard $\text{ETS}(E,T,S)$ notation, corresponding to the formulation of the \textbf{E}rror, \textbf{T}rend, and \textbf{S}easonal components respectively. The algorithms consistently isolate non-seasonal configurations with additive errors ($\text{ETS}(A,A,N)$ variants). Given that the dynamic factors are strongly trended, the automated framework selectively implements damped trends ($\text{ETS}(A,A_d,N)$) to avoid explosive long-term extrapolations. A damped configuration is selected in 100\% of female windows and 80\% of male windows. Similarly, under the TWA-MFTS framework, a damped trend specification is automatically preferred in 74.3\% (713 out of 960) of the model fits.
\end{itemize}

\end{document}